\documentclass[preprint,12pt]{elsarticle}

\usepackage{amsmath,amsthm,amsfonts,amssymb,graphicx,color,verbatim,hyperref}
\numberwithin{equation}{section}

\biboptions{square}

\newtheorem{theorem}{Theorem}
\newtheorem{proposition}[theorem]{Proposition}
\newtheorem{corollary}[theorem]{Corollary}
\newtheorem{lemma}[theorem]{Lemma}

\journal{Journal of Theoretical Biology }

\def\I{\mathbb{ I}}
\def\i{\textit}
\def\Y{\begin{array}{c}	y_1 \\ \vdots \\ y_n \end{array}}
\def\O{\begin{array}{c}	1 \\ \vdots \\ 1 \end{array}}
\def\I_n{
\begin{array}{ccc}1 & \, & \, \\ \, & \ddots & \, \\ \, & \, & 1 \end{array}
}
\def\Z{\begin{array}{c} 0 \\ \vdots \\ 0 \end{array}}
\def\YVn^t{\begin{array}{cccc} \Y & \Y & \ldots & \Y \end{array}}
\def\YVan^t{\begin{array}{cccc} \Z & \Y & \ldots & \Y \end{array}}
\def\YVbn^t{\begin{array}{cccc} \Y & \Z & \ldots & \Y \end{array}}
\def\YVnn^t{\begin{array}{cccc} \Y & \Y & \ldots & \Z \end{array}}
\def\Za{\begin{array}{cccc} \Y & \Z & \ldots & \Z \end{array}}
\def\Zb{\begin{array}{cccc} \Z & \Y & \ldots & \Z \end{array}}
\def\Zn{\begin{array}{cccc} \Z & \Z & \ldots & \Y \end{array}}

\def\OYVn^t{\begin{array}{cccc} \O & \O & \ldots & \O \end{array}}
\def\OYVan^t{\begin{array}{cccc} \Z & \O & \ldots & \O \end{array}}
\def\OYVbn^t{\begin{array}{cccc} \O & \Z & \ldots & \O \end{array}}
\def\OYVnn^t{\begin{array}{cccc} \O & \O & \ldots & \Z \end{array}}
\def\OZa{\begin{array}{cccc} \O & \Z & \ldots & \Z \end{array}}
\def\OZb{\begin{array}{cccc} \Z & \O & \ldots & \Z \end{array}}
\def\OZn{\begin{array}{cccc} \Z & \Z & \ldots & \O \end{array}}

\newcommand{\diag}{\text{diag}}

\begin{document}

\begin{frontmatter}

\title{Reduced models of networks of coupled enzymatic reactions}

\author[a]{Ajit Kumar}
\ead{ajitkmar@math.uh.edu}

\author[a,b]{Kre\v{s}imir Josi\'{c}\corref{c}}
 \address[a]{Department of Mathematics, University of Houston, Houston TX 77204-3008,  USA}
 \address[b]{Department of Biology and Biochemistry, University of Houston, Houston TX 77204-5001,  USA}
\cortext[c]{Corresponding author }
\ead{josic@math.uh.edu}

\begin{abstract}
The Michaelis-Menten equation has played a central role in our understanding of biochemical processes.  It has long been understood how this equation approximates the dynamics of irreversible enzymatic reactions.    However, a similar approximation in the case of networks, where the product of one reaction can act as an enzyme in another, has not been fully developed. 
Here we rigorously derive such an approximation in a class of coupled enzymatic networks where the individual interactions are of Michaelis-Menten type. 
 We show that the sufficient conditions for the validity of the \textit{total quasi steady state assumption (tQSSA)}, obtained in a  single protein case by Borghans, de Boer and Segel
 can be extended to   sufficient conditions for the validity of the tQSSA in a large class of enzymatic networks. 
Secondly,  we derive reduced equations that approximate the network's dynamics and involve only  protein concentrations.  This significantly reduces the number of equations necessary to model such systems. 
We prove the validity of this approximation using geometric singular perturbation theory and results about matrix differentiation.   The ideas used in deriving the approximating equations are quite general, and can be used to systematize other model reductions.  
\end{abstract}


\end{frontmatter}



\section{Introduction}

The Michaelis-Menten (MM) scheme~\citep{BH,MM} is a fundamental building block of many models of
 protein interactions:
An enzyme, $E$,  reacts with a protein, $X$, resulting in an intermediate complex, $C$.  In turn, this complex can break down into a product, $X_p$, and the enzyme $E$. It is frequently assumed that formation of $C$ is reversible while its breakup  is not.   The process is represented by the following sequence of reactions~\citep{BH,MM,murray2003}
\begin{equation} \label{MM}
X + E \overset{k_1}{\underset{k_{-1}}{\rightleftharpoons}} C \overset{k_2}{\rightarrow} X_p + E.
\end{equation} 

Frequently catalytic activity in protein interaction network is modeled by MM equations~\citep{huangFerrell1996,novakPatakiCilibertoTyson2001,novakPatakiCilibertoTyson2003,UriAlon2007,davitichBornholdt2008,stadman1977,NovakTyson1993,goldbeter91}. This gives rise to \emph{coupled enzymatic networks}, where the  substrate of one reaction acts as enzyme in another reaction. 
A direct application of the law of mass action to such models
 typically leads to high dimensional differential equations which are often stiff, and difficult to study directly.  
 
 A number of methods have been introduced to address these problems. Most of these methods are based on \emph{quasi steady state assumptions} which take advantage of the  differences in characteristic timescales of the  quantities  being modeled.  It is typically assumed that the chemical species, or some combinations of chemical species, can be divided into two classes: One which equilibrates rapidly, and  a second which evolves more slowly~\cite{hek2010,othmerLee2010}.  Assuming that the members of the first class equilibrate instantaneously  leads to a reduced model involving only elements of the second class.  
 
 Reduction methods differ in their assumption on which chemical species, or combinations thereof, are assigned to the two different classes.   For instance, the  \emph{standard quasi steady state assumption (sQSSA)} posits that the concentrations of the intermediate complexes change quickly compared to the protein concentration~\citep{GK,segel,segelAndSlemrod,maini,NoethenWalcher2006}. An alternative is the \emph{reverse quasi steady state assumption (rQSSA)} where the protein concentration is assumed to change rapidly compared to intermediate complexes~\citep{SchnellMaini2000}. Rigorous justifications of these methods are largely available only for  isolated reactions of the type shown in scheme~\eqref{MM}, and the Goldbeter-Koshland switch\footnote{A Goldbeter-Koshland switch consist of two coupled reactions.  One of these reactions frequently represents protein phosphorylation, and the second  dephosphorylation.}~\citep{GK}. 
 
 The total quasi steady state assumption (tQSSA) was introduced to broaden the range of parameters over which a \emph{quasi steady state assumption} is valid.  Under this assumption  the concentration of the intermediate complex, $C$,  evolves quickly compared   to the sum of the intermediate complex and the protein concentration~\citep{BorghansBoerSegel1996,tzafriri,hegland,PedersenBersaniBersani2008,tzafririEdelman2004}.
Numerical experiments and heuristic arguments suggest that tQSSA may be valid in coupled enzymatic networks over a very broad set of parameters~\citep{CilibertoFabrizioTyson2007}.

Here we aim to provide a theoretical foundation for the reductions used in  numerical studies of enzymatic networks.  A standard model reduction technique for systems involving quantities that change on different timescales is geometric singular perturbation theory (GSPT)~\citep{Fenichel1979,kaper,hek2010}. For instance, this theory has been used by Khoo and Hegland to prove several results obtained earlier by Borghans, et al. using self consistency arguments~\citep{hegland,BorghansBoerSegel1996}.
 GSPT has also been used to reduce other models of biochemical reactions~\citep{ZagarisKaperKaper2004,HardinZagarisKrabWesterhoff2009,othmerLee2010}. We derive a sufficient condition for the validity of tQSSA in arbitrary  networks of  proteins and enzymes provided the interactions are of MM type and  can be modeled by mass action kinetics. This directly extends previous work, like that of Pedersen, et al.~\citep{PedersenBersaniBersaniCortese2008} who proposed a sufficient condition for the validity of tQSSA in the Goldbeter-Koshland switch. 

The direct application of the tQSSA to coupled enzymatic networks generally leads to a differential-algebraic system. The algebraic part of this system consists of coupled quadratic equations that are typically impossible to solve.  Our second aim is to show that, under certain assumptions on the structure of the network, it is possible to circumvent this problem using
ideas introduced by Bennett, et al.~\citep{BennettVolfsonTsimringHasty2007}.  This allows us to obtain reduced set of differential equations for a class of protein interaction networks in terms of protein concentrations only. 

We proceed as follows: In section~\ref{IsolatedMichaelisMentenReaction} we review the original Michaelis--Menten scheme.  We introduce terminology, and illustrate our approach in a simple setting. In
this section we also give a brief overview of the theory of geometric singular perturbation theory, which is fundamental in proving the validity of the reduced equations.  In section~\ref{sec:may3010_TwoProtein} we extend our  approach to a well studied two protein network that plays part in the G2-to-mitosis phase (G2/M) transition in the eukaryotic cell cycle. We present the ideas in the most general setting in section~\ref{sec:TheGeneralProblem}, where we derive the general form of the reduced equations. Each section  begins by the discussion of the tQSSA in the context of the network under consideration, and  closes with a derivation of the reduced equations under the tQSSA, as well as sufficient conditions under which the tQSSA holds.  A number of technical details used in the proofs of the main results are given in the appendices. 
We note that  throughout the presentation the law of mass action is assumed to hold.

\section{Isolated Michaelis-Menten reaction}\label{IsolatedMichaelisMentenReaction}

The MM scheme is frequently used to model enzymatic processes in solution
which are ubiquitous in biology.   As discussed in the introduction, a number of different
approaches have been proposed to justify the reduced equations mathematically.   We start by giving
a detailed overview of the tQSSA approach based on geometric singular perturbation theory
(GSPT)\citep{Fenichel1979}. The setting of a single MM type reaction will be used to introduce the main ideas and difficulties of reducing equations that describe larger reaction networks.

For notational convenience we will use variable names to denote both a chemical species and its concentration.  For instance, $E$  denotes both an enzyme and its concentration.
Reaction~\eqref{MM} reaction obeys two natural  constrains: The total amount of protein and enzyme remain constant.  Therefore,
\begin{equation}\label{25Mar10_XT}
X + C + X_p = X_T, \quad \text{and} \quad 
E + C = E_T, 
\end{equation}
for positive constants $X_T$ and $E_T$.
In conjunction with the constraints~\eqref{25Mar10_XT},  the following system of ordinary differential equations 
can be used to model reaction~\eqref{MM}
\begin{align}\label{25mar10_oneDimEq}
\frac{dX}{dt} &= -k_1X(E_T-C) +k_{-1}C,          &X(0) &= X_T, \notag \\
\frac{dC}{dt} &= k_1X(E_T-C) - (k_{-1}+k_2)C, & C(0)  &= 0.  
\end{align}

\subsection{The total quasi steady state assumption (tQSSA)}
\label{S:intro_one}

Under the standard quasi steady state assumption (sQSSA), the concentration of the substrate--bound enzyme, $C$, equilibrates quickly, which allows 
system~\eqref{25mar10_oneDimEq} to be reduced by one dimension.     Sufficient conditions under
which the sQSSA is valid have been studied extensively~\citep{GK,segel,maini}. However, it has also  been  observed that
the sQSSA is too restrictive~\citep{BorghansBoerSegel1996,tzafriri}.   

To obtain a reduction that is valid for a wider range of parameters, define  $ \bar{X} := X+C$. Eq.~(\ref{25mar10_oneDimEq}) can then be rewritten as
\begin{subequations}\label{26mar10_A}
\begin{align}
\frac{d\bar{X}}{dt} &= -k_2C, 
&\bar{X}(0) &= X_T, \label{26mar10_Aa} \\ 
\frac{dC}{dt} &= k_1[\bar{X}E_T-(\bar{X}+E_T  + k_m)C +C^2], &C(0) &= 0,\label{26mar10_Ab}
\end{align}
\end{subequations}
where $k_m = (k_{-1}+k_2)/k_1$ is the \textit{Michaelis--Menten constant}. 

The tQSSA posits that $C$ equilibrates quickly compared to $\bar{X}$~\citep{BorghansBoerSegel1996,tzafriri}.
Under this assumption we obtain the following differential--algebraic system 
\begin{subequations}\label{27mar10_A}
\begin{align}
\frac{d\bar{X}}{dt} &= -k_2C, &\qquad
\bar{X}(0) = X_T,\label{27mar10_Aa} \\
                  0 &= k_1[\bar{X}E_T-(\bar{X}+E_T  + k_m)C +C^2].\label{27mar10_Ab} 
\end{align}
\end{subequations}
Solving Eq.~(\ref{27mar10_Ab}) and noting that only the negative branch of solutions is stable,  we can express $C$ in terms of $\bar{X}$ to obtain a closed, first order differential equation for $\bar{X}$,
\begin{equation}\label{24May10_barX}
 \frac{d \bar{X}}{dt} = -k_2\frac{
(\bar{X}+E_T  + k_m) - \sqrt{(\bar{X}+E_T  + k_m)^2- 4\bar{X}E_T} } {2}
, \qquad
\bar{X}(0) = X_T.
\end{equation}
Although the reduced equation is given in the $\bar{X}, C$ coordinates, it is easy to revert to the original variables $X,C$.
Therefore, from Eq.~(\ref{24May10_barX}) one can recover an approximation to the solution of Eq.~(\ref{25mar10_oneDimEq}). 

\subsection{Extension of the tQSSA}
\label{S:extension}

An essential step in the tQSSA reduction is the solution of the quadratic equation~(\ref{27mar10_Ab}). A direct  extension  of this approach to networks of chemical reaction typically leads to coupled system of quadratic equations~\citep{CilibertoFabrizioTyson2007,PedersenBersaniBersaniCortese2008,PedersenBersaniBersani2008}. 
The solution of this system may not be unique, and generally needs to be obtained numerically.  
However, an approach introduced by Bennett, et al.~\citep{BennettVolfsonTsimringHasty2007}, can be used to obtain the desired solution from a system of linear equations.  

In particular, we keep the
tQSSA, but look for a reduced equation in the original 
coordinates, $X, C$.  
Using  $ \bar{X} = X + C$ to eliminate $\bar{X}$ from Eq.~(\ref{27mar10_Ab}), we obtain
\begin{equation}\label{28april10_C}
 0 = k_1\left(X(E_T-C) - k_mC\right).
\end{equation}
Eq.~(\ref{28april10_C}) and Eq.~(\ref{27mar10_Ab}) 
are equivalent, but  Eq.~(\ref{28april10_C}) is linear in $C$, and leads to 
\[
 C = \frac{XE_T}{k_m+X},  \text{ and } \bar{X} = X + \frac{XE_T}{k_m+X}.
\]
Using these expressions formulas in Eq.~(\ref{27mar10_Aa}), and applying the chain rule gives
\begin{subequations}
\begin{equation} 
\frac{\partial}{\partial X} \left( X + \frac{XE_T}{k_m+X}\right)\frac{dX}{dt} = -k_2 \frac{XE_T}{k_m+X} 
\quad \Longrightarrow \quad
\frac{dX}{dt} = -k_2 \left( 1 + \frac{k_mE_T}{(k_m+X)^2}\right)^{-1} \frac{XE_T}{k_m+X}. \label{27mar10_dxdt}
\end{equation}

The reduced Eq.~\eqref{27mar10_dxdt} was obtained  under the assumption that there is no significant change in $\bar{X} = X + C$ during the  rapid equilibration. After equilibration,  $C = XE_T/(k_m+X)$ (See Fig.~\ref{fig:initialValue}). Therefore, the initial value for Eq.~(\ref{27mar10_dxdt}), denote by  $\hat{X}(0)$, can be obtained from the initial values  $X(0), C(0)$ using
\begin{equation}\label{june0810_initialVal}
 \hat{X}\left(0\right) + \frac{E_T\hat{X}\left(0\right)}{\hat{X}\left(0\right)+k_m} = X(0)+ C(0)  =  X_T.
\end{equation}
\end{subequations}

Fig.~\ref{fig:initialValue}c)  shows that the solutions of  the full system~(\ref{25mar10_oneDimEq})  and the reduced system~(\ref{27mar10_dxdt}) are close when initial conditions are mapped correctly.

\clearpage
\begin{figure}[t] 
\begin{center}
\includegraphics[scale=.165]{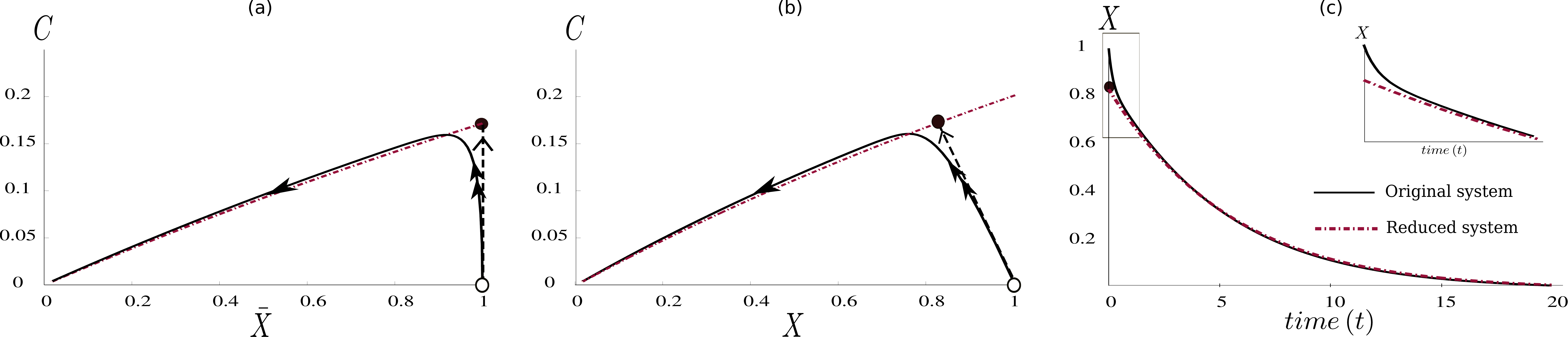}
\end{center}
\caption{\footnotesize{
Proper choice of the initial values of the reduced system. The empty circle at $ \bar{X} = 1,\, C = 0,$ represents the initial value for  the full system. The solid dot is the initial value of the reduced system.  The dash-dotted (red) line represents the attracting \emph{slow manifold}. \emph{(a)}  The solid curve represents the numerical solution of Eq.~(\ref{26mar10_A}).     The solution rapidly converges to the manifold, and evolves slowly along the manifold after this transient. The dashed line satisfies $\bar{X} = X_T$. The solid dot at the intersection of the dashed line and the slow manifold represents the projection of the initial condition onto the slow manifold given by Eq.~(\ref{27mar10_Ab}).  Thus $\bar{X}(0) = X_T$ is the proper initial condition for the reduced system~\eqref{24May10_barX}.  \emph{(b)}  The solid line represents the numerical solution of Eq. (\ref{25mar10_oneDimEq}).  After a quick 
transient, the solution again converges to the slow manifold.  However, since the initial
transient is not orthogonal to the $X$ axis, the initial conditions do not project vertically onto the slow manifold.  Instead, the initial transient follows the line $X + C = X_T$ (dashed), and   the intersection of this  line and the slow manifold represents the proper choice of the initial value for Eq.~(\ref{27mar10_dxdt}). 
\emph{(c)}
Comparison of solutions of Eq.~(\ref{25mar10_oneDimEq}) and the reduced system~(\ref{27mar10_dxdt}). The graph in the inset offers a magnified view of the boxed region, showing the quick transient to the slow manifold. 
We used: $X_T = E_T = k_1 = k_2 = 1, \, k_{-1} = 3$, which, using 
 Eq.~(\ref{june0810_initialVal}), gives the initial condition for the reduced system, $\hat{X}(0) = 0.83$. } }
 \label{fig:initialValue}
\end{figure}
\clearpage

The tQSSA  implies that Eq.~\eqref{26mar10_A} can be approximated by Eq.~\eqref{27mar10_A}.  
Therefore, to explore the conditions under which
 Eq.~(\ref{27mar10_dxdt}) is a valid  reduction of Eq.~(\ref{25mar10_oneDimEq}) we need to 
provide the asymptotic limits under which the transition from  Eq.~\eqref{26mar10_A} to Eq.~\eqref{27mar10_A} is justified. Different sufficient conditions for the tQSSA have been obtained   using self-consistency arguments~\citep{BorghansBoerSegel1996,tzafriri}.  We follow the ideas of  \i{pairwise balance} to look for a proper non-dimensionalisation of variables~\citep{segelAndSlemrod,maini}. Although this method gives a weaker result than the one obtained in~\citep{tzafriri}, it is 
easier to extend to networks of reactions. 

\subsection{Review of Geometric singular perturbation (GSPT)}
\label{S:fenichel}

Since geometric singular perturbation theory (GSPT) is essential in our 
reduction of the equations  describing
coupled enzymatic reactions, we here provide a very brief overview of the theory.  Further details can be found in~\citep{Fenichel1979,wiggins1994,jones1995,kaper,hek2010}.  Readers familiar with
GSPT can skip to section~\ref{tQSSA_valid}. 

Consider a system of ordinary differential equation of the form
\begin{subequations}\label{E:group}
\begin{align}\label{10may10_GSPT}
 \epsilon \frac{du}{dt} &= f(u,v,\epsilon), & u(0) &= u_0, \nonumber \\
          \frac{dv}{dt} &= g(u,v,\epsilon), & v(0) &= v_0, 
\end{align}
where $u \in \mathbb{R}^k $ and $v \in \mathbb{R}^l $ with $k,l \ge 1$, and $u_0 \in \mathbb{R}^k $, $v_0 \in \mathbb{R}^l $ are initial values. The parameter $\epsilon$ is assumed to be small and positive $(0 < \epsilon \ll 1)$, the  functions $f$ and $g$  smooth, 
\begin{equation} \label{E:cond1}
f(u,v,0) \not\equiv 0,   g(u,v,0) \not\equiv 0, \quad \text{and} \quad \lim_{\epsilon \rightarrow 0} \epsilon g(u,v,\epsilon) \equiv 0.
\end{equation}
\end{subequations}
The variable $u$ is termed the \emph{fast} variable, and $v$  the \emph{slow} variable.

Assume that $\mathcal{M}_0 := \{(u,v) \in \mathbb{R}^{k+l} \, |\, f(u,v,0) = 0 \}$ is a compact, smooth manifold with inflowing boundary. Suppose further that the eigenvalues $\lambda_i$ of the Jacobian $\frac{\partial f}{\partial u}(u,v,0)|_{\mathcal{M}_0}$ all satisfy $Re(\lambda_i) < 0$,  so that $\mathcal{M}_0$ is \emph{normally hyperbolic}.
Then, for $\epsilon$ sufficiently small, the solutions of Eq.~(\ref{10may10_GSPT}) follow an initial transient, which can be approximated by 
\begin{align}\label{may3010_GSPT_Breakb}
          \frac{du}{ds} &= f(u,v,0) , & u(0) = u_0, \nonumber \\
          \frac{dv}{ds} &= 0        , & v(0) = v_0,
\end{align}
where $t = \epsilon s$.  After this transient,  the solutions are $\mathcal{O}(\epsilon)$ close to the solutions of the reduced system
\begin{align}
                     0  &= f(u,v,0),     \nonumber  \\
          \frac{dv}{dt} &= g(u,v,0), \quad v(0) = v_0.  \label{may3010_GSPT_Breaka}
\end{align} 
 More precisely  there is an invariant, slow manifold $\mathcal{M}_\epsilon$,  $\mathcal{O}(\epsilon)$ close to $\mathcal{M}_0$. Solutions of Eq.~(\ref{10may10_GSPT}) are attracted to $\mathcal{M}_0$ exponentially fast,  and can be approximated by concatenating the fast transient described by
 Eq.~\eqref{may3010_GSPT_Breakb}, and the solution of the reduced Eq.~\eqref{may3010_GSPT_Breaka}.

 The slow manifold, $\mathcal{M}_0$, consists of the fixed points of Eq.~(\ref{may3010_GSPT_Breakb}). The condition that the eigenvalues, $\lambda_i$,  of the Jacobian $\frac{\partial f}{\partial u}(u,v,0)|_{\mathcal{M}_0}$ all satisfy $Re(\lambda_i) < 0$ implies that these fixed points  are stable. 

\subsection{Validity of the tQSSA}\label{tQSSA_valid}
\label{S:validity1}

We next show that GSPT can be applied to Eq.~(\ref{26mar10_A}), after a suitable rescaling of  variables~\citep{segelAndSlemrod,maini}.
Let
\begin{equation}\label{newVariables}
\tau = \frac{t}{T_{\bar{X}}}, \quad \bar{x}(\tau) = \frac{\bar{X}(t)}{X_T}, \quad c(\tau) = \frac{C(t)}{\beta}. 
\end{equation}
We have some freedom in defining $\beta$ and $T_{\bar{X}}$.  Using 
 the method of \emph{pairwise balance}~\citep{maini,segelAndSlemrod}, we let
\begin{equation}\label{may2610_betaAndTxBar}
 \beta = \frac{X_TE_T}{X_T+E_T  + k_m}, \quad \text{and} \quad  T_{\bar{X}} = \frac{X_T}{k_2 \beta}.
\end{equation}
In the rescaled variables, Eq.~(\ref{26mar10_A}) 
takes the form
\begin{subequations}\label{8may10_B}
\begin{align}
\frac{d\bar{x}}{d\tau} &= -  c, 
&\bar{x}(0)&=1,\label{8may10_Ba} \\
\frac{k_2}{k_1}\frac{E_T}{(E_T+X_T+k_m)^2} \frac{ dc }{d\tau} &= \bar{x}-\frac{X_T\bar{x}+E_T  + k_m}{X_T+E_T  + k_m} c +\frac{X_TE_T}{(E_T+X_T+k_m)^2}c^2, &c(0)&=0. \label{8may10_Bb}
\end{align}
\end{subequations}
Define the parameter
\begin{equation}\label{27mar10_eps}
 \epsilon := \frac{\beta}{k_1T_{\bar{X}}X_TE_T}  = \frac{k_2}{k_1}\frac{E_T}{(E_T+X_T+k_m)^2}  .
\end{equation}
For small $\epsilon$,  Eq.~\eqref{8may10_B} is singularly perturbed and has the form given in Eq.~(\ref{10may10_GSPT}). Indeed, we can apply GSPT to Eq.~\eqref{8may10_B} directly since in
the limit $\epsilon \rightarrow 0$ the right hand side of Eq.~(\ref{8may10_B}) remains ${\mathcal O}(1)$.  Indeed, the requirement $0 < \epsilon \ll 1$, is  equivalent to the sufficient condition for the validity of the tQSSA  derived in~\citep{BorghansBoerSegel1996}. 

GSPT implies that for small $\epsilon$,  solutions of 
Eq.~(\ref{8may10_B}) are close to those of the reduced system 
\begin{subequations} \label{may2610_rescaled}
\begin{align}
\frac{d\bar{x}}{d\tau} &= -c,&\bar{x}(0)&=1, \label{may2610_rescaled_a}\\
                     0 &=  \, \bar{x}-\frac{X_T\bar{x}+E_T  + k_m}{X_T+E_T  + k_m} c +\frac{X_TE_T}{(E_T+X_T+k_m)^2}c^2 \label{may2610_rescaled_b}.
\end{align}
\end{subequations}

The normal hyperbolicity and stability of the manifold defined by Eq.~(\ref{may2610_rescaled_b}) can be verified directly, and also follow from the results of section~\ref{sec:TheGeneralProblem}.   It follows that GSPT can be applied to conclude that GSPT implies that Eq.~(\ref{may2610_rescaled}) is a reduction of Eq.~(\ref{8may10_B}). 


The validity of the reduction in these rescaled equations implies its validity in the original coordinates:   Eq.~(\ref{8may10_B}) is equivalent to  Eq.~(\ref{26mar10_A}) via the scaling given in Eq.~(\ref{newVariables}).    Hence, Eq.~(\ref{27mar10_A}) and Eq.~(\ref{may2610_rescaled}) are related by the same scaling relationship.  We note that choosing the initial values of the intermediate complexes to be zero, implies that solutions of~\eqref{8may10_B} remain ${\mathcal O}(1)$ for small $\epsilon$ (see section \ref{Sect.ZeroInitComplex} for a detailed discussion).
It follows that
Eq.~(\ref{27mar10_A}) is a valid reduction of Eq.~(\ref{26mar10_A}) when $\epsilon$ is  sufficiently small.
Hence, for $\epsilon$ in the same range, Eq.~(\ref{27mar10_dxdt}), with initial values satisfying Eq.~(\ref{june0810_initialVal}), is a valid reduction of Eq.~(\ref{25mar10_oneDimEq}).

Lemma~\ref{lemma1} in the Appendix shows that $\epsilon$ is always smaller than $1/4$.  Although
this is suggestive, GSPT only guarantees the validity of the reduced equations in some unspecified range
of $\epsilon$ values.

\section{Analysis of a two protein network}\label{sec:may3010_TwoProtein}

We next show how the reduction described in the previous section extends to a network of MM reactions.  Here the substrate of one reaction acts as an enzyme in another reaction. To illustrate the main ideas used in reducing the corresponding equations,
we start with a concrete example of two interacting proteins.  

\clearpage
\begin{figure}[t]
\begin{center}
\includegraphics[scale=.19]{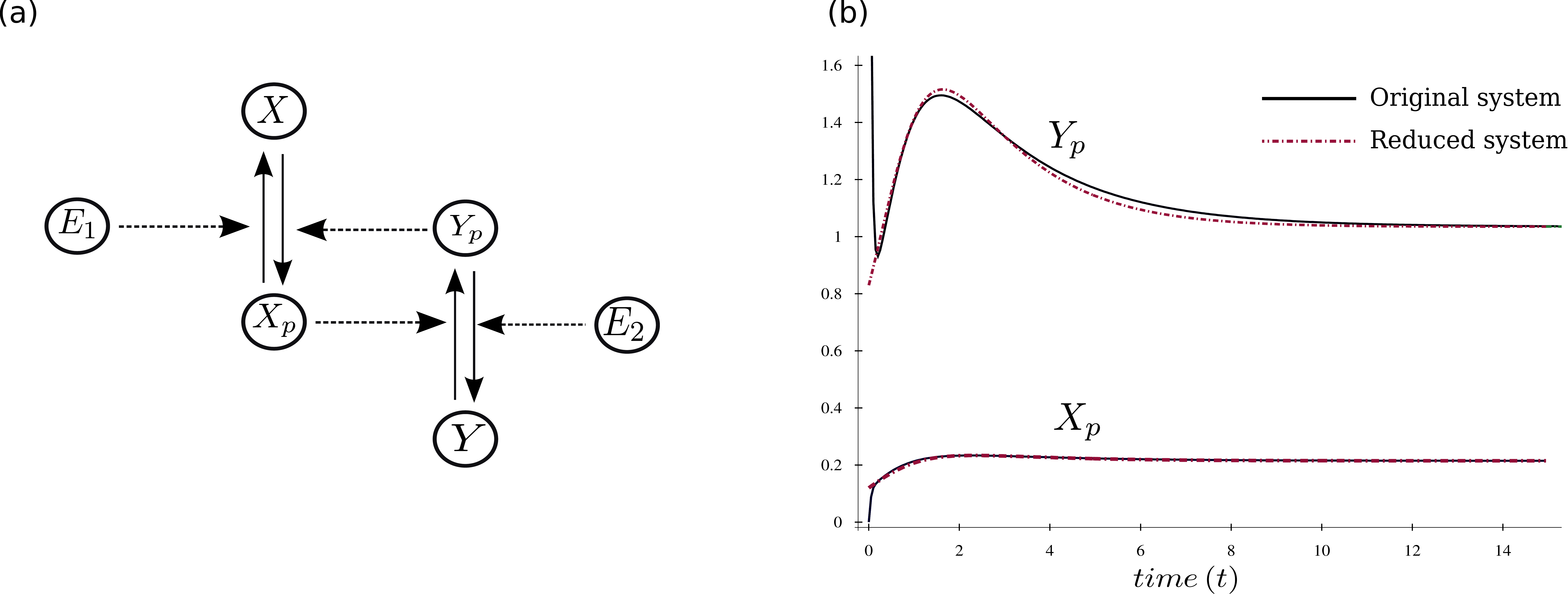}
\end{center}
\caption{\footnotesize{ A simplified description of interactions between two regulators of the G2-to-mitosis phase (G2/M) transition in the eukaryotic cell cycle \citep{NovakTyson1993} (See text). \emph{(a)} $X$ and $Y$ phosphorylate and deactivate each other. For instance, the protein
  $X$ exists in a \textit{phosphorylated} $X_p$ and \textit{unphosphorylated} $X$  state, and the conversion $X $ to $ X_p$ is catalyzed by $Y_p$.   The conversion of $X_p$ to $X$ is catalyzed by the phosphatase $E_1$.   \emph{(b)} 
Comparison of the numerical solution of Eq.~(\ref{27feb10_ode6dim}) and  Eq.~(\ref{02Mar10_finalSolution}). Here  $k_1 = 5, k_{-1}=1, k_{2} = 1, E_1^T = 10, E_2^T = 2, X_T = 10, Y_T = 10.1$ as in~\citep{CilibertoFabrizioTyson2007}. The initial values for  Eq.~(\ref{27feb10_ode6dim}) are $X(0) = 10, Y(0) = 1.1, X_p(0) = 0, Y_p(0) = 9, C_x(0) = 0, C_y(0) = 0, C_x^e(0) = 0, C_y^e(0) = 0, E_1(0) = 10, E_2(0) = 2$. The initial values of the reduced system, $\hat{X}_p(0) = 0.12, \hat{Y}_p(0) = 0.83$ are obtained by the projection onto the slow manifold defined by Eq.~\eqref{E:projection}.
} }
\label{25feb10_fig1}
\end{figure}
\clearpage

Fig.~\ref{25feb10_fig1}a)   is a simplified depiction of the interactions between two regulators of the G2-to-mitosis phase (G2/M) transition in the eukaryotic cell cycle \citep{NovakTyson1993}. Here, $Y$ represents \emph{MPF} (M-phase promoting factor, a dimer of \emph{Cdc2} and cyclin B) and $X$ represents \emph{Wee1} (a kinase that phosphorylates and deactivates \emph{Cdc2}). The proteins
  exist in a \textit{phosphorylated} state, $X_p,Y_p$, and an \textit{unphosphorylated} state, $X,Y$, with the phosphorylated state being less active.  The proteins $X$ and $Y$ deactivate each other, and hence act as antagonists. In this network $E_1$ and $E_2$ represent phosphatases that catalyze the conversion of $X_p$ and $Y_p$ to $X$ and $Y,$ respectively. Each dotted arrow  in Fig.~\ref{25feb10_fig1}a) is associated with exactly one  MM type reaction in the list of reactions given below. The sources of the arrows act as enzymes.  Therefore, Fig. ~\ref{25feb10_fig1}a) represents the following network of reactions
\begin{eqnarray*}
Y_p+X \overset{k_1}{\underset{k_{-1}}{\rightleftarrows}} C_x \overset{k_2}{\longrightarrow}   X_p +Y_p, && 
E_1+X_p \overset{k_1}{\underset{k_{-1}}{\rightleftarrows}} C_x^e \overset{k_2}{\longrightarrow}  X+E_1 , \\
X_p+Y \overset{k_1}{\underset{k_{-1}}{\rightleftarrows}} C_y \overset{k_2}{\longrightarrow} Y_p+ X_p, &&
E_2+Y_p \overset{k_1}{\underset{k_{-1}}{\rightleftarrows}} C_y^e \overset{k_2}{\longrightarrow}   Y+E_2.
\end{eqnarray*}
To simplify the exposition, we have assumed some homogeneity in the rates.  
Since the total concentration of proteins and enzymes is assumed  fixed, the 
system obeys the following set of  constraints 
\begin{align*}
 X_T &= X(t) + X_p(t) + C_x(t) + C_y(t)+C_x^e(t) ,  
&E_1^T &=  C_x^e(t) + E_1(t) , \\
Y_T &=  Y(t) + Y_p(t) + C_x(t) + C_y(t)+ C_y^e(t) , 
&E_2^T  &=  C_y^e(t) + E_2(t),
\end{align*}
where $X_T, Y_T, E_1^T, E_2^T$ are constant and represent the total concentrations of the respective 
 proteins and enzymes.   Along with these constraints the concentrations
 of  the ten species
in the reaction evolve according to
 \begin{align}
\frac{dX_p}{dt} 
&= -k_1\underbrace{(Y_T-Y_p-C_x-C_y-C_y^e)}_{= Y}X_p-k_1X_p\underbrace{(E_1^T-C_x^e)}_{= E_1} 
+k_{-1}C_x^e+(k_{-1}+k_2)C_y+k_2C_x,
\notag
\\
\frac{dY_p}{dt} 
&= -k_1\underbrace{(X_T-X_p-C_x-C_y-C_x^e)}_{=X}Y_p-k_1Y_p\underbrace{(E_2^T-C_y^e)}_{=E_2} 
+k_{-1}C_y^e+(k_{-1}+k_2)C_x+k_2C_y, \notag
\\
\frac{dC_x}{dt} &= k_1\underbrace{(X_T-X_p-C_x-C_y-C_x^e)}_{=X}Y_p-(k_{-1}+k_2)C_x,  \label{27feb10_ode6dim}\\
\frac{dC_y}{dt} &=k_1\underbrace{(Y_T-Y_p-C_x-C_y-C_y^e)}_{= Y}X_p-(k_{-1}+k_2)C_y, \notag \\
\frac{dC_x^e}{dt} &= k_1X_p\underbrace{(E_1^T-C_x^e)}_{= E_1}-(k_{-1}+k_2)C_x^e,  \notag\\
\frac{dC_y^e}{dt} &=k_1Y_p\underbrace{(E_2^T-C_y^e)}_{=E_2}-(k_{-1}+k_2)C_y^e, \notag
\end{align}
with initial values 
\begin{equation} \label{E:initial1}
 C_x(0) = 0, \quad C_y(0) = 0,\quad C_x^e(0) = 0,\quad C_y^e(0) = 0.
\end{equation}
The initial values of $X_p$  and $Y_p$ are arbitrary. 

Following the approach in the previous section, we  reduce Eq.~\eqref{27feb10_ode6dim} to a two dimensional system.  Assuming the validity of the tQSSA, we obtain an approximating differential--algebraic system.  Solving the algebraic equations, which are linear in the original coordinates, leads to a closed, reduced system of ODEs. We  end by discussing the validity of the tQSSA.
    
\subsection{New coordinates and reduction under the tQSSA}
\label{S:intro_two}

To extend the tQSSA we define a new set of variables by adding the concentration of the free state of a species to the concentrations of all intermediate complexes formed by that particular species as reactant~\citep{CilibertoFabrizioTyson2007}, 
\begin{equation}\label{08feb10_slowVar2}
 \begin{array}{ccc} 
\bar{X}_p &:=& X_p  + C_y + C_x^e, \vspace{1mm}\\
\bar{Y}_p &:=& Y_p  + C_x + C_y^e.
 \end{array}
\end{equation}

%
%
Under the tQSSA, the 
intermediate complexes equilibrate quickly compared to the variables
$\bar{X}_p$ and
$\bar{Y}_p$. In the coordinates defined by Eq.~\eqref{08feb10_slowVar2}, Eq.~\eqref{27feb10_ode6dim} takes the form
\begin{subequations}\label{03mar10_TQSSA}
\begin{align}
\frac{d\bar{X}_p}{dt} &= k_2C_x -k_2C_x^e, \label{03mar10_TQSSAa}\\
\frac{d\bar{Y}_p}{dt} &=k_2C_y - k_2C_y^e, \label{03mar10_TQSSAb}\\
0 &= k_1(X_T-\bar{X}_p-C_x)(\bar{Y}_p-C_x -C_y^e)-(k_{-1}+k_2)C_x,  \label{03mar10_TQSSAc}\\
0 &=k_1(Y_T-\bar{Y}_p-C_y)(\bar{X}_p-C_y -C_x^e)-(k_{-1}+k_2)C_y, \label{03mar10_TQSSAd} \\
0 &= k_1(\bar{X}_p-C_y -C_x^e)(E_1^T-C_x^e)-(k_{-1}+k_2)C_x^e,  \label{03mar10_TQSSAe}\\
0 &=k_1(\bar{Y}_p-C_x -C_y^e)(E_2^T-C_y^e)-(k_{-1}+k_2)C_y^e. \label{03mar10_TQSSAf}
\end{align}
\end{subequations}
Solving the coupled system of quadratic equations~(\ref{03mar10_TQSSAc}-\ref{03mar10_TQSSAf}) in terms of $\bar{X}_p,\bar{Y}_p$ appears to be possible only numerically, as it is equivalent to 
finding the roots of a degree 16 polynomial~\citep{CilibertoFabrizioTyson2007}. 
However, since we are interested in the dynamics of $X_p $ and  $ Y_p$, we can 
proceed as in the previous section: Using Eq.~ (\ref{08feb10_slowVar2}) in    (\ref{03mar10_TQSSAc}-\ref{03mar10_TQSSAf})  gives a linear system in $C_x,C_y,C_x^e,C_y^e$.  Defining $k_m := (k_{-1}+k_2)/k_1$, this 
system  can be written in matrix form  as
\begin{equation}\label{2mar10_manifoldMatrix2}
 \left[\begin{array}{cccc}
        Y_p+k_m   &   Y_p          & Y_p          &         0        \\
        X_p           &   X_p+k_m  &      0        &    X_p         \\
              0         &         0        & X_p+k_m &         0        \\
              0         &         0        &       0        &   Y_p+k_m
       \end{array}
 \right]
\left[\begin{array}{c}
        C_x \\ C_y \\C_x^e \\ C_y^e
       \end{array}
 \right]
=
\left[\begin{array}{c}
        Y_p(X_T-X_p) \\ X_p(Y_T-t_p) \\ X_pE_1^T \\ Y_pE_2^T
       \end{array}
 \right].
\end{equation}
The coefficient matrix above is invertible and  
 Eq.~(\ref{2mar10_manifoldMatrix2}) can be solved to obtain $C_x,C_y,C_x^e,C_y^e$ as functions of $X_p,Y_p$.  Denoting  the resulting solutions as  $C_x(X_p,Y_p),$ $ C_y(X_p,Y_p),$ $C_x^e(X_p,Y_p),$ $C_y^e(X_p,Y_p)$ and using them  in Eqs.(\ref{03mar10_TQSSAa}-\ref{03mar10_TQSSAb}) we obtain the closed system of equations
\begin{eqnarray}
 \frac{d}{dt}\left[\begin{array}{c} 
                    \bar{X}_p \\ \bar{Y}_p
                   \end{array}
 \right] &=& k_2\left[\begin{array}{c} 
                    C_x(X_p,Y_p) - C_x^e(X_p,Y_p) \\ C_y(X_p,Y_p)-C_y^e(X_p,Y_p)
                   \end{array}
 \right]. \notag 
 \end{eqnarray}
Reverting to the original coordinates, $X_p $ and $ Y_p$, and using the chain rule gives
\begin{eqnarray}
\frac{d}{dt}\left[\begin{array}{c} 
                    X_p+C_y(X_p,Y_p)+C_x^e(X_p,Y_p) \\ Y_p+C_x(X_p,Y_p)+C_y^e(X_p,Y_p)
                   \end{array}
 \right] &=& k_2\left[\begin{array}{c} 
                    C_x(X_p,Y_p) - C_x^e(X_p,Y_p) \\ C_y(X_p,Y_p)-C_y^e(X_p,Y_p)
                   \end{array}
 \right] \Longrightarrow  \notag \\
\left[\begin{array}{cc}
       1 + \frac{\partial C_y }{\partial X_p }+\frac{\partial C_x^e  }{\partial X_p } & \frac{\partial C_y }{\partial Y_p } + \frac{\partial C_x^e}{\partial Y_p} \\
  \frac{\partial C_x }{\partial X_p } +  \frac{\partial C_y^e }{\partial X_p } &1 + \frac{\partial C_x}{\partial Y_p } +\frac{\partial C_y^e }{\partial Y_p}
      \end{array}
 \right]
\frac{d}{dt}\left[\begin{array}{c} 
                    X_p \\ Y_p
                   \end{array}
 \right] &=& k_2\left[\begin{array}{c} 
                    C_x(X_p,Y_p) - C_x^e(X_p,Y_p) \\ C_y(X_p,Y_p)-C_y^e(X_p,Y_p)
                   \end{array}
 \right]. \label{may2810_final}
\end{eqnarray}

The initial values of Eq.~(\ref{may2810_final}) are determined by projecting the initial values, given by  Eq.~(\ref{E:initial1}),  onto the slow manifold. Unfortunately, they can be expressed only implicitly. The reduction from Eq.~(\ref{27feb10_ode6dim}) to Eq.~(\ref{may2810_final}) was obtained under the assumption that $\bar{X}_p = X_p + C_y + C_x^e$ and $\bar{Y}_p = Y_p + C_x + C_y^e$ are slow variables, and hence constant during the transient to the slow manifold. Therefore the projections of the initial conditions onto the slow manifold, $\hat{X}_p(0)$ and
$\hat{Y}_p(0)$, are related to the original initial conditions
as
\begin{equation} \label{E:projection}
\begin{array}{ccccccc}
\hat{X}_p(0)  + C_y(\hat{X}_p(0) ,\hat{Y}_p(0)) + C_x^e(\hat{X}_p(0),\hat{Y}_p(0))\quad &=&  X_p(0) + C_y(0) + C_x^e(0) \quad&=& X_T, \\
\hat{Y}_p(0) + C_x(\hat{X}_p(0),\hat{Y}_p(0)) + C_y^e(\hat{X}_p(0),\hat{Y}_p(0))\quad &=&  Y_p(0) + C_x(0) + C_y^e(0) \quad&=& Y_T.
\end{array}
\end{equation}
 
 We have therefore shown that,  if the tQSSA holds, and if the coefficient matrix on the left hand side of Eq.~(\ref{may2810_final}) is invertible, then 
\begin{eqnarray}\label{02Mar10_finalSolution}
\frac{d}{dt}\left[\begin{array}{c} 
                    X_p \\ Y_p
                   \end{array}
 \right] &=& k_2 \left[\begin{array}{cc}
       1 + \frac{\partial C_y }{\partial X_p }+\frac{\partial C_x^e  }{\partial X_p } & \frac{\partial C_y }{\partial Y_p } + \frac{\partial C_x^e}{\partial Y_p} \\
  \frac{\partial C_x }{\partial X_p } +  \frac{\partial C_y^e }{\partial X_p } &1 + \frac{\partial C_x}{\partial Y_p } +\frac{\partial C_y^e }{\partial Y_p} 
      \end{array}
 \right]^{-1}\left[\begin{array}{c} 
                    C_x(X_p, Y_p) - C_x^e(X_p, Y_p) \\ C_y(X_p, Y_p)-C_y^e(X_p, Y_p)
                   \end{array}
 \right],   
\end{eqnarray}
with initial value obtained by solving Eq.~(\ref{E:projection}), 
is a valid approximation of Eq.~(\ref{27feb10_ode6dim}).
Fig.~\ref{25feb10_fig1}b) shows that the solutions of the two systems are indeed close, after an initial transient.

\subsection{Validity of the tQSSA for two interacting proteins}
\label{S:validity2}

To reveal the asymptotic limits for which the tQSSA holds, we again 
rescale  the original equations. In particular,  $\bar{X}_p$ and $\bar{Y}_p$ are scaled by the total concentration of the respective proteins. To scale the intermediate complexes, each MM reaction in this network is treated as isolated. The scaling factors are then obtained analogously to $\beta$ in Eq.~(\ref{may2610_betaAndTxBar}). Let
\begin{align*}
\alpha_x &:= \frac{X_TY_T}{X_T+Y_T+k_m} ,  &\alpha_y &:=\frac{X_TY_T}{X_T+Y_T+k_m},      \\
\beta_x^e &:= \frac{X_TE_1^T}{X_T+E_1^T+k_m} , &\beta_y^e &:=\frac{Y_TE_2^T}{Y_T+E_2^T+k_m},
\end{align*}
and
\begin{equation*}
 T_s := \text{max} \left\{\frac{X_T}{k_2\alpha_x},\frac{X_T}{k_2\beta_x^e} ,\frac{Y_T}{k_2\alpha_y},\frac{Y_T}{k_2\beta_y^e} \right\}.
\end{equation*}
Therefore, $T_s$ is obtained analogously to $T_{\bar{X}}$ in Eq.~(\ref{may2610_betaAndTxBar}).
The reason for choosing the maximum will become evident shortly.
The rescaled variables are now defined as
\begin{eqnarray}\label{11may10_newVar}
 \tau := \frac{t}{T_s},\quad \bar{x}_p(\tau) &:=& \frac{\bar{X}_p(t)}{X_T}, \quad \bar{y}_p(\tau) := \frac{\bar{Y}_p(t)}{Y_T}, \notag\\
c_x(\tau) := \frac{C_x(t)}{\alpha_x} ,\quad  c_y(\tau) := \frac{C_y(t)}{\alpha_y}, &\quad&
c_x^e(\tau) := \frac{C_x^e(t)}{\beta_x^e} , \quad c_y^e(\tau) := \frac{C_y^e(t)}{\beta_y^e}.  
\end{eqnarray}

Using Eq.~(\ref{08feb10_slowVar2}) in the Eq.~(\ref{27feb10_ode6dim}) to eliminate $X_p, Y_p$, and then 
applying the rescaling, defined by Eq.~ (\ref{11may10_newVar}), to the new ODE  we obtain 
\begin{subequations}\label{27feb10_ode6dimNewVarRescaled}
\begin{align} 
\frac{d\bar{x}_p}{d\tau} 
&=
 \frac{k_2\alpha_xT_s}{X_T}c_x -\frac{k_2\beta_x^eT_s}{X_T}c_x^e, \label{27feb10_ode6dimNewVarRescaleda}
\\
\frac{d\bar{y}_p}{d\tau} 
&=
 \frac{k_2\alpha_yT_s}{Y_T}c_y -\frac{k_2\beta_y^eT_s}{Y_T}c_y^e, \label{27feb10_ode6dimNewVarRescaledb}
\\
\underbrace{\frac{\alpha_x}{k_1X_TY_TT_s}}_{\le \epsilon_x}\frac{dc_x}{d\tau} 
&=\begin{array}{l} \big[
\bar{y}_p - \bar{x}_p \bar{y}_p - \frac{\alpha_x }{X_T} c_x \bar{y}_p - \frac{\alpha_x }{Y_T}c_x - \frac{\beta_y^e }{Y_T}c_y^e + \frac{ \alpha_x }{Y_T}c_x \bar{x}_p 
+ \frac{\beta_y^e }{Y_T}c_y^e \bar{x}_p + \frac{\alpha_x^2 }{X_T Y_T}c_x^2 
\\
+ \frac{ \alpha_x \beta_y^e }{X_T Y_T}c_x c_y^e - \frac{\alpha_x  k_m}{X_T Y_T}c_x \big], \end{array} \label{27feb10_ode6dimNewVarRescaledc} 
\\
\underbrace{\frac{\alpha_y}{k_1X_TY_TT_s}}_{\le \epsilon_y}\frac{dc_y}{d\tau} 
&= \begin{array}{l} \big[
\bar{x}_p - \frac{\beta_x^e }{X_T}c_x^e - \frac{\alpha_y }{X_T}c_y - \bar{x}_p \bar{y}_p + \frac{\beta_x^e }{X_T}c_x^e \bar{y}_p + \frac{ \alpha_y }{X_T}c_y \bar{y}_p 
- \frac{\alpha_y }{Y_T}c_y \bar{x}_p + \frac{\alpha_y \beta_x^e }{X_T Y_T}c_x^e c_y 
\\
+ \frac{ \alpha_y^2 }{X_T Y_T}c_y^2 - \frac{\alpha_y  k_m}{X_T Y_T}c_y \big], \end{array} \label{27feb10_ode6dimNewVarRescaledd} 
\\
\underbrace{\frac{\beta_x^e}{k_1X_TE_1^TT_s}}_{\le \epsilon_x^e}\frac{dc_x^e}{d\tau} 
&=
\bar{x}_p - \frac{\beta_x^e }{E_1^T}c_x^e \bar{x}_p - \frac{\beta_x^e }{X_T}c_x^e - \frac{\alpha_y }{X_T}c_y + \frac{(\beta_x^e)^2 }{ E_1^T X_T}(c_x^e)^2 
+ \frac{\alpha_y \beta_x^e }{E_1^T X_T}c_x^e c_y - \frac{\beta_x^e  k_m}{E_1^T X_T}c_x^e ,
 \label{27feb10_ode6dimNewVarRescalede}
\\
\underbrace{\frac{\beta_x^e}{k_1E_2^TY_TT_s}}_{\le \epsilon_y^e}\frac{dc_y^e}{d\tau} 
&=
\bar{y}_p - \frac{\beta_y^e }{E_2^T}c_y^e \bar{y}_p - \frac{\alpha_x }{Y_T}c_x -\frac{\beta_y^e }{Y_T}c_y^e + \frac{\alpha_x \beta_y^e }{ E_2^T Y_T}c_x c_y^e  
+ \frac{(\beta_y^e)^2 }{E_2^T Y_T}(c_y^e)^2 - \frac{\beta_y^e  k_m}{E_2^T Y_T}c_y^e ,
\label{27feb10_ode6dimNewVarRescaledf}
\end{align}
\end{subequations}
where\begin{align*}
 \epsilon_x &:= \frac{k_2}{k_1} \frac{Y_T}{(X_T + Y_T + k_m)^2}, &\epsilon_y &:= \frac{k_2}{k_1} \frac{X_T}{(Y_T + X_T + k_m)^2}, \\
 \epsilon_x^e &:= \frac{k_2}{k_1} \frac{E_1^T}{(X_T + E_1^T + k_m)^2},  &\epsilon_y^e &:=  \frac{k_2}{k_1} \frac{E_2^T}{(Y_T + E_2^T + k_m)^2}.
 \end{align*}
The bounds on these coefficients follow from the definition of $T_s$. Since
$
({1}/{T_s} )\le ({k_2 \alpha_x}/{X_T}),
$

\[
\frac{\alpha_x}{k_1X_TY_TT_s} 
\le
\frac{k_2}{k_1} \frac{\alpha_x^2}{X_T^2 Y_T} = \frac{k_2}{k_1} \frac{1}{X_T^2 Y_T} \left(\frac{X_TY_T}{X_T+Y_T+k_m}\right)^2 = \epsilon_x.
\]
Similarly, 
\[
\frac{\alpha_y}{k_1X_TY_TT_s} \le \epsilon_y, \quad
\frac{\beta_x^e}{k_1X_TE_1^TT_s} \le \epsilon_x^e, \quad
\text{and}
\quad 
\frac{\beta_x^e}{k_1E_2^TY_TT_s} \le \epsilon_y^e.
\]
Finally, we define 
\begin{equation}\label{4may10_epsilon}
 \epsilon := \max\left\{ \epsilon_x,  \epsilon_y,\epsilon_x^e, \epsilon_y^e 
 \right\}.
\end{equation}
The definitions of scaling factors in (\ref{11may10_newVar}) imply that all the coefficients on the right hand side of (\ref{27feb10_ode6dimNewVarRescaledc}--\ref{27feb10_ode6dimNewVarRescaledf}) are $\mathcal{O}(1)$. Therefore, in the asymptotic limit $\epsilon \rightarrow 0$, Eq.~(\ref{27feb10_ode6dimNewVarRescaled}) defines a singularly perturbed system. Since the two equations 
are related by the scaling given in Eq.~(\ref{11may10_newVar}), we can conclude that in the limit $\epsilon \rightarrow 0$, the tQSSA is valid.  If additionally the slow manifold is normally hyperbolic, then  Eq.~(\ref{03mar10_TQSSA}) is a valid reduced model of the network's dynamics.  The normal hyperbolicity and stability of the slow manifold will be proved in a general setting in section~\ref{sec:TheGeneralProblem} .
%

\section{The general problem}\label{sec:TheGeneralProblem} 

We next describe how to obtain reduced equations describing the dynamics of a large class of protein interaction networks~\citep{huangFerrell1996,novakPatakiCilibertoTyson2001,novakPatakiCilibertoTyson2003,UriAlon2007,davitichBornholdt2008,stadman1977,NovakTyson1993,goldbeter91}.
We again assume that the proteins interact via MM type reactions, and that a generalization of the tQSSA holds~\citep{CilibertoFabrizioTyson2007}. We will follow the steps that lead to the reduced systems in the previous two sections: After describing the model and the conserved quantities, we recast the equations in terms of the ``total'' protein concentrations (\emph{cf.} sections~\ref{S:intro_one} and \ref{S:intro_two}).    Under a generalized tQSSA, these equations can be reduced 
to an algebraic-differential system.  We show that the algebraic part of the system is linear  in the 
original coordinates (\emph{cf.} sections~\ref{S:extension} and~\ref{S:intro_two}), so that the reduced system can be described by a differential equation with dimension equal to the
number of interacting proteins.  We next show that this reduction is justified by proving that the singularly perturbed system we examine satisfies the conditions of GSPT (\emph{cf.} section~\ref{S:fenichel}).  Finally, we describe  
the asymptotic conditions under which the system is singularly perturbed,
following the arguments in sections~\ref{S:validity1} and~\ref{S:validity2}.

\subsection{Description of the network}\label{june2110_setup}

We start by defining the nodes and edges of a general protein interaction  network. The nodes in this network represent enzymes as well as proteins, while the edges represent the catalytic effect  one species has on another. Proteins are assumed to come in two states, phosphorylated and unphosphorylated.  Both states are represented by a single node in this network.  Fig.~\ref{fig:june1510_nNode} and the following description make these definitions 
precise.

\clearpage
\begin{figure}[t]
\begin{center}
\includegraphics[scale = .15]{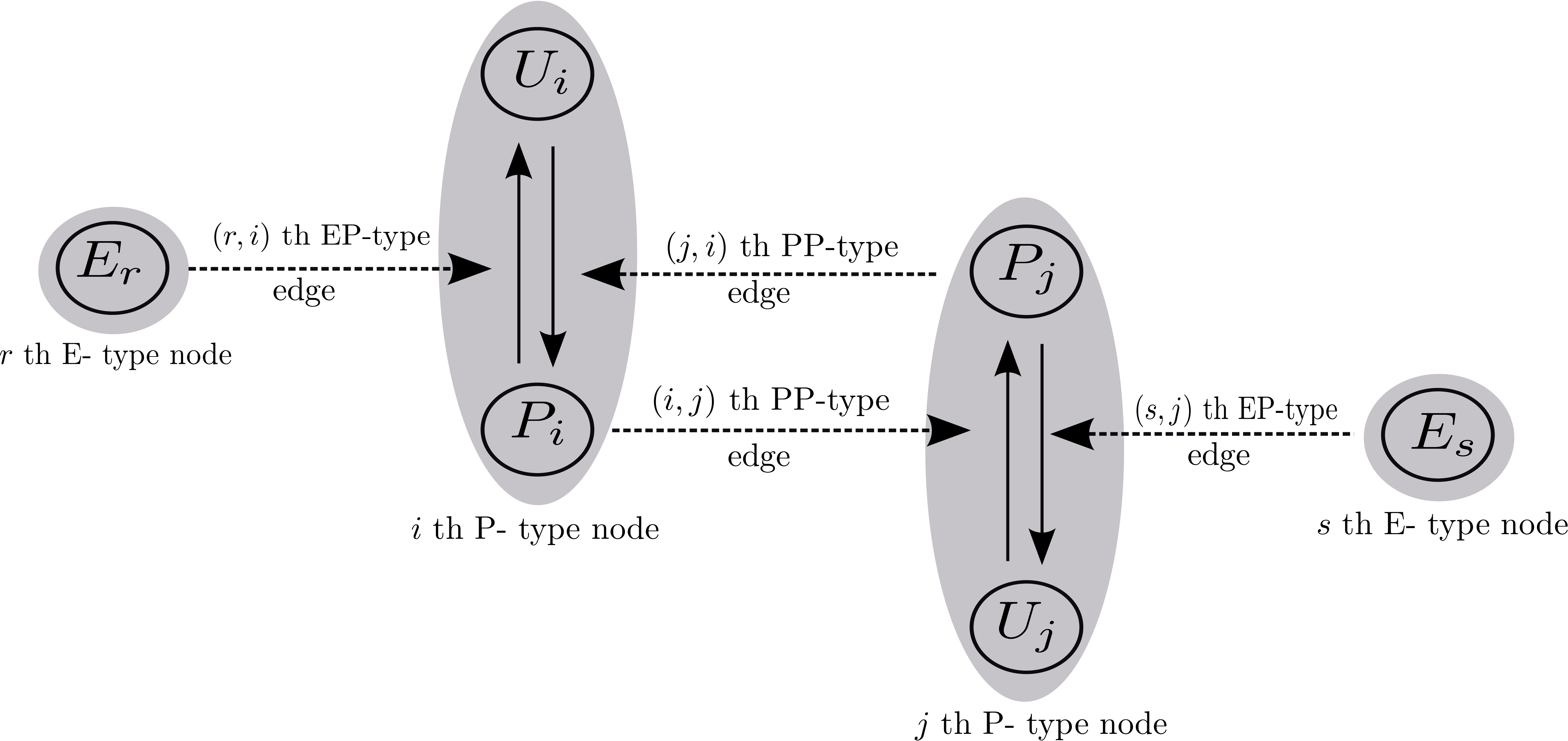}
\end{center}
\caption{ \footnotesize{ A simple example illustrating the terminology used in describing protein interaction networks. The shaded regions represent nodes and encompass either an enzyme or a single protein that is part of an MM type reaction.  Each dotted arrow represents an edge in the network. The  solid arrows represent transitions within the nodes, and do not define an edge in the network. } }
\label{fig:june1510_nNode}
\end{figure}
\clearpage

In a network of  $n$ interacting proteins, and $n$ associated enzymes, we define the following: 

\textbf{Nodes:} The two types of nodes in this network represent proteins (P-type nodes)  and  enzymes (E-type nodes).  Each protein can exist in either an \emph{active} or \emph{inactive} form. The \emph{inactive} form of the $i$th protein is denoted by $U_i$, and the \emph{active} form by $P_i$.  The $i$th P-type node is formed by grouping together  $U_i$ and $P_i$. In addition there  are $n$ species of enzymes, $E_i$, which exist in only one state. 

\textbf{Edges:} All edges in the network are \emph{directed}, and represent the catalytic effect of a species in a MM type reaction. There are two types of edges: \i{PP-type} edges connect two P-type nodes, while \i{EP-type}  edges connect E-type nodes to P-type nodes.  
In particular, a PP-type edge from node $i$ to node $j$ represents the following MM type reaction in which $P_i$ catalyzes the conversion of $U_j$ to the \emph{active} form $P_j$,
\begin{subequations}
\begin{equation}\label{may3010_PPij}
P_i+U_j \overset{k^{1}_{ij}}{\underset{k^{-1}_{ij}}{\rightleftarrows}} C^U_{ij} \overset{k^2_{ij}}{\longrightarrow} P_j + P_i.
\end{equation}
Note that autocatalysis is possible. 
 The rate constants  $k^{1}_{i,j},k^{-1}_{i,j},k^{2}_{i,j}$,  associated to each edge, can be grouped into weighted``connectivity matrices''
\begin{equation*}
\begin{array}{ccccccccc}
K_1 &=& \left[k^1_{ij}\right]_{n\times n},&
K_{-1} &=& \left[k^{-1}_{ij}\right]_{n\times n}, &
K_2 &=& \left[k^2_{ij}\right]_{n\times n}\end{array}.
\end{equation*}
In the absence of an edge, that is, when
$P_i$ does not catalyze the phosphorylation of  $U_j$, the corresponding
$(i,j)$-th entry in $K_1,K_{-1},$ and $K_2$ is set to zero.

EP-type edges are similar to PP-type edges, with  enzymes acting as catalysts. To  each pair of enzyme, $E_i$, and protein, $P_j$, we associate three rate constants $l^{1}_{i,j},l^{-1}_{i,j},l^{2}_{i,j}$ of the corresponding reaction in which $E_i$  is a catalyst in the conversion of $P_j$ into $U_j$,
\begin{equation}\label{may3010_EPij}
E_i+P_j\overset{l^{1}_{ij}}{\underset{l^{-1}_{ij}}{\rightleftarrows}} C^E_{ij} \overset{l^2_{ij}}{\longrightarrow} U_j +E_i.
\end{equation}
\end{subequations}
  The rate constants can again be arranged into matrices
\begin{equation*}
\begin{array}{ccccccccc}
L_1 &=& \left[l^1_{ij}\right]_{n\times n},&
L_{-1} &=& \left[l^{-1}_{ij}\right]_{n\times n}, &
L_2 &=& \left[l^2_{ij}\right]_{n\times n},
\end{array}
\end{equation*}
with zero entries again denoting the absence of interactions. 

These definitions  imply that the active form of one protein always
catalyzes the production of the active form of another  protein. 
This assumption excludes certain interactions (see section~\ref{sec:discussion} for an example).  
However, the reduction is easiest to describe under these assumptions, and 
we discuss generalizations in the Discussion. 

For notational convenience we define 
$U = [ U_1,  U_2, \ldots,  U_n ]^t, 
P = [ P_1,  P_2, \ldots,  P_n ]^t,$ and 
$E =[ E_1,  E_2, \ldots,  E_n ]^t$,  
and arrange intermediate complexes into matrices,
\begin{equation*}\label{20mar10_cxcy}
C_U = \left[\begin{array}{cccc} C_{11}^U \\ C_{21}^U \\ \vdots \\ C_{n1}^U      
            \end{array}
\begin{array}{cccc} C_{12}^U \\ C_{22}^U \\ \vdots \\ C_{n2}^U      
            \end{array}
 \ddots 
\begin{array}{cccc} C_{1n}^U \\ C_{2n}^U \\ \vdots \\ C_{nn}^U      
            \end{array}
\right],
\quad
C_E = \left[\begin{array}{cccc} C_{11}^E \\ C_{21}^E \\ \vdots \\ C_{n1}^E      
            \end{array}
\begin{array}{cccc} C_{12}^E \\ C_{22}^E \\ \vdots \\ C_{n2}^E      
            \end{array}
 \ddots 
\begin{array}{cccc} C_{1n}^E \\ C_{2n}^E \\ \vdots \\ C_{nn}^E      
            \end{array}
\right].
\end{equation*}
Initially all intermediate complexes are assumed to start at zero concentration.  Therefore, 
any intermediate complex corresponding to a reaction that has zero rates, will remain at zero
concentration for all time. 

\begin{quote}
For instance, in the two protein example analyzed in section~\ref{sec:may3010_TwoProtein},  we have
\end{quote}
\[\
 C_U =  \left[ \begin{array}{c} 0 \\ C_y \end{array} 
               \begin{array}{c}C_y  \\ 0  \end{array}
               \right], \quad
C_E =  \left[ \begin{array}{c} C_x^e \\ 0 \end{array} 
               \begin{array}{c} 0 \\ C_y^e  \end{array}
               \right], \quad
 U = \left[ \begin{array}{c}X \\ Y \end{array} \right], \quad
P = \left[ \begin{array}{c}X_p \\ Y_p \end{array} \right], \quad
 \text{ etc.}
\]


Assuming that the system is isolated from the environment implies that the total concentration of each enzyme, $E_i^T$,  remains constant. Therefore,
\begin{subequations}\label{may3110_constraints}
\begin{equation}\label{20mar10_ET}
 E_i + \sum_{s=1}^n{C_{is}^E} = E_i^T ,\qquad i \in \{1,2,...,n \}.
\end{equation}
Similarly, for each protein the total concentration,  $U_i^T$, of its \emph{inactive} and \emph{active} form, and the intermediate complexes is constant, 
\begin{equation}\label{20mar10_UT}
 U_i + P_i + \left( \sum_{s=1}^nC_{is}^U + \sum_{r=1}^nC_{ri}^U  - C^U_{ii} \right) + \sum_{r=1}^nC_{ri}^E  = U_i^T, \qquad i \in \{1,2,...,n \}.
\end{equation}
\end{subequations}
Let 
$$
V_n = \underbrace{[\begin{array}{cccc}1 & 1 &  \ldots & 1 \end{array}]^t}_{n \text{ times}}, \quad 
E_T  = [ \begin{array}{cccc}E^T_1 & E^T_2 & \ldots & E^T_n \end{array}]^t,  \quad \text{and} \quad
U_T = [ \begin{array}{cccc} U^T_1 &  U^T_2 & \ldots & U^T_n \end{array}]^t,
$$
and denote  the $n \times n$ identity matrix by $I_n$.
In addition, we use  the \textit{Hadamard product} of matrices, denoted by $*$,  to simplify notation\footnote{For instance, the Hadamard product of matrices $A = \left[\begin{array}{cc}a & b \\ c &d \end{array}\right],$ and $ B = \left[\begin{array}{cc}e & f \\ g &h \end{array}\right],\text{ is }  A*B = \left[\begin{array}{cc}ae & bf \\ cg &dh \end{array}\right].$}. Constraints~(\ref{may3110_constraints}) can now be written concisely in matrix form 
\begin{align*}
E_T &= E+ C_E V_n ,                                            \\
U_T &= U + P +C_U V_n + C_U^t V_n- (I_n*C_U)V_n + C_E^t V_n.
\end{align*}
Applying the law of mass action to the system of reactions described by~(\ref{may3010_PPij}-\ref{may3010_EPij}) yields a $(2n^2+n)$ dimensional dynamical system, 
\begin{align}
 \frac{dP_i}{dt} &= \sum_{s=1}^n \bigg(-k_{is}^1 P_iU_s + ( k_{is}^{-1}+ k_{is}^{2})C_{is}^U \bigg) +\sum_{r=1}^n\bigg( k_{ri}^2C_{ri}^U 
  - l_{ri}^1 E_r P_i +  l_{ri}^{-1} C_{ri}^E\bigg),
& P_i(0) &= p_{i}^0,
\nonumber \\
 \frac{dC_{ij}^U}{dt} &= \, k_{ij}^1 P_iU_j-(k_{ij}^{-1}+k_{ij}^{2})C_{ij}^U,
& C_{ij}^U(0) &= 0. 
\label{nov0809_mainODE} \\
\frac{dC_{ij}^E}{dt} &= \, l_{ij}^1E_iP_j -(l_{ij}^{-1}+l_{ij}^{2})C_{ij}^E, 
& C_{ij}^E(0) &= 0, 
\notag
\end{align}
Due to the constraints~(\ref{20mar10_ET},\ref{20mar10_UT}),  $U_i, E_i $, are  \textit{affine linear} function of $P_i,C_{ij}^U, C_{ij}^E$ and can be used to close  Eq.~(\ref{nov0809_mainODE}). Our aim is to  reduce this $2n^2+n$ dimensional system to an $n$ dimensional system involving only  $P_i$.

\subsection{The total substrate coordinates}
\label{S:general_coordinates}

In this section we generalize the  change of variables to the ``total`` protein concentrations, introduced in Eq.~(\ref{08feb10_slowVar2}).  Let
\begin{equation}\label{may3110_PiBarTermWise}
 \bar{P}_i := P_i + \sum_{s=1}^nC_{is}^U + \sum_{r=1}^nC_{ri}^E, \quad i \in \{1,2,...,n\},
\end{equation}
so that Eq.~(\ref{nov0809_mainODE}) takes the form
\begin{subequations}\label{may3010_ODEtermWise}
\begin{align}
 \frac{d\bar{P}_i}{dt} =&  \sum_{r=1}^n k_{ri}^{2} C_{ri}^U  -\sum_{r=1}^nl_{ri}^2 C_{ri}^E,   \label{may3010_ODEtermWisea}\\
\frac{dC_{ij}^U}{dt} =& \, k_{ij}^1P_iU_j-(k_{ij}^{-1}+k_{ij}^{2})C_{ij}^U,  \label{may3010_ODEtermWisec} \\
\frac{dC_{ij}^E}{dt} =&\, l_{ij}^1E_iP_j -(l_{ij}^{-1}+l_{ij}^{2})C_{ij}^E. \label{may3010_ODEtermWiseb} 
\end{align}
\end{subequations}
To close this system we use Eqs.~(\ref{20mar10_ET},\ref{20mar10_UT}) with Eq.~(\ref{may3110_PiBarTermWise}), to obtain
\begin{eqnarray}\label{21mar_xieiyi}
U_i &=& U_i^T - P_i - \sum_{s=1}^nC_{is}^U - \sum_{r=1}^n\left(C_{ri}^U + C_{ri}^E \right) + C_{ii}^U \nonumber \\
    &=& U_i^T - \bar{P}_i -   \sum_{r=1}^nC_{ri}^U+ C_{ii}^U, \nonumber \\
E_i &=& E_i^T -\sum_{s=1}^n{C_{is}^E}, \\
P_i &=& \bar{P}_i - \sum_{s=1}^nC_{is}^U - \sum_{r=1}^nC_{ri}^E. \nonumber
\end{eqnarray}
Defining $\bar{P} := (\bar{P}_1,\bar{P}_2,...,\bar{P}_n)^t$,  Eq.~(\ref{may3110_PiBarTermWise}) can be written in vector form as $ \bar{P} = P+C_UV_n+C_E^tV_n$, and 
Eqs.~(\ref{may3010_ODEtermWise}) and (\ref{21mar_xieiyi}) can be written in matrix form as
\begin{subequations}\label{09mar10_matrixForm}
\begin{align}
\frac{d\bar{P}}{dt} =& (K_2* C_U)^t V_n - (L_2* C_E)^tV_n, \label{09mar10_matrixForma} \\
\frac{dC_U}{dt} =& K_1* (P U^t) -(K_{-1}+K_2)* C_U, \label{09mar10_matrixFormc} \\
\frac{dC_E}{dt} =& L_1* (EP^t)- (L_{-1}+L_2)* C_E, \label{09mar10_matrixFormb} 
\end{align}
\end{subequations}
where
\begin{subequations}\label{09mar10_closedform}
\begin{align}
U  =&\,  U_T- P -C_U V_n - C_U^t V_n - C_E^t V_n + (I_n*C_U)V_n  \nonumber \\
=&\, U_T - \bar{P} - C_U^t V_n+ (I_n*C_U)V_n,  \label{09mar10_closedformb} \\
E =&\, E_T- C_E V_n, \label{09mar10_closedforma} \\
P =&\,  \bar{P} -C_UV_n-C_E^tV_n.  \label{09mar10_closedformc}
\end{align}
\end{subequations}

\subsection{The  tQSSA and the resulting reduced equations}
\label{S:general_reduced}


The general form of the tQSSA states that the intermediate complexes, $C_U$ and $C_E$, equilibrate
faster than  $\bar{P}$. This assumption implies that, after a fast transient, Eq.~(\ref{09mar10_matrixForm}) can be approximated by the differential-algebraic system
\begin{subequations}\label{09mar10_matrixFormDEL}
\begin{align}
\frac{d\bar{P}}{dt} =& (K_2* C_U)^t V_n - (L_2* C_E)^tV_n, \label{09mar10_matrixFormDELa} \\
0 =& K_1* (P U^t) -(K_{-1}+K_2)* C_U, \label{09mar10_matrixFormDELc} \\
0 =& L_1* (EP^t)- (L_{-1}+L_2)* C_E. \label{09mar10_matrixFormDELb}
\end{align}
\end{subequations}
In particular,  according to GSPT (see section~\ref{S:fenichel}), if the  \emph{slow manifold}
\begin{equation}\label{20june10_slowManifold}
\mathcal{M}_0 = \left\{\left(\bar{P},C_U,C_E \right )  \, \bigg| \, 
\begin{array}{ccl}
0 &=& K_1* (P U^t) -(K_{-1}+K_2)* C_U; \\
0 &=& L_1* (EP^t)- (L_{-1}+L_2)* C_E 
\end{array}
\right\}
\end{equation}
is normally hyperbolic and stable, then the solutions of Eq.~(\ref{09mar10_matrixForm}) are
attracted to and shadow solutions on $\mathcal{M}_0$.

If we consider the system~(\ref{09mar10_matrixFormDEL}b,c) entry-wise then it consists of $2n^2$  coupled quadratic equations in $2n^2 + n$ variables, namely the entries of $\bar{P},C_U,C_E$ (note that $U,E$ are functions of $\bar{P},C_U,C_E$). 
As described in section~\ref{S:intro_two}, we can avoid solving coupled quadratic equations by seeking a solution in terms of $P$ instead of $\bar{P}$. Using Eq.~(\ref{09mar10_closedform}a,b) we eliminate $E,U$ from Eqs.~(\ref{09mar10_matrixFormDEL}b,c) to obtain
\begin{subequations}\label{10mar10_linear}
\begin{align}
K_1* \left[P \left(  V_n^tC_U^t + V_n^tC_U- V_n^t (I_n * C_U)\right) +  PV_n^t C_E  \right] +(K_{-1}+K_2)* C_U 
&=\quad K_1* \left[P \left( U_T^t- P^t\right)\right], \label{10mar10_lineara}\\
L_1* \left(C_E \left(V_n  P^t\right)\right)+ (L_{-1}+L_2)* C_E  &= \quad L_1* \left(E_TP^t\right). 
\label{10mar10_linearb}
\end{align}
\end{subequations}
Although complicated, Eq.~(\ref{10mar10_linear}) is linear in $C_U$ and $C_E$. The following Lemma, proved in ~\ref{june0110_solveLinearEquation}, shows that the equations are also solvable.
\begin{lemma}\label{21june10_linerSolve}
Suppose 
$K_1 = [k_{ij}^{1}],$ $K_{-1}= [k_{ij}^{-1}],$ $K_2= [k_{ij}^{2}],$ $L_1= [l_{ij}^{1}],$ $L_{-1}= [l_{ij}^{-1}],$ $L_2= [l_{ij}^{2}]$ $\in \mathbb{R}^{n \times n}$ are real matrices with non-negative entries.  Furthermore,  assume that  for any pair $i,j \in \{1,2,...,n\} $  either $k_{ij}^{1} = k_{ij}^{-1} = k_{ij}^{2} = 0 $, or all these coefficients are positive, and similarly for the coefficients $l_{ij}^{1},  l_{ij}^{-1}$, and $ l_{ij}^{2}$.
If $U_T,E_T, P $ $\in\mathbb{R}^{n \times 1}_+ $ are real vectors with positive entries, and $V_n = [1\,1\,\cdots\,1]^t$ is a vector of size $n$, then Eq.~(\ref{10mar10_linear}) has a unique solution for $C_U, C_E \in \mathbb{R}^{n \times n} $ in terms of $P$.
\end{lemma}


We denote the solution of Eq.~(\ref{10mar10_linear}) described in  Lemma~\ref{21june10_linerSolve} by  $\tilde{C}_U(P), \tilde{C}_E(P)$.
This solution can be used to close Eq.~(\ref{09mar10_matrixFormDELa}), by using Eq.~\eqref{09mar10_closedformc} to obtain
\begin{eqnarray}\label{19mar10_B}
\frac{d\bar{P}}{dt} &=&  \frac{dP}{dt} + \frac{d}{dt}\left(\tilde{C}_U(P)V_n\right) + \frac{d}{dt}\left(\tilde{C}_E(P)^tV_n\right) \nonumber\\
&=&\left[I +\frac{\partial}{\partial P}\left(\tilde{C}_U(P)V_n\right) + \frac{\partial}{\partial P}\left(\tilde{C}_E(P)^tV_n\right) \right]\frac{dP}{dt}
\end{eqnarray}
With Eq.~(\ref{09mar10_matrixFormDELa}), this leads to a closed system in $P$,
\begin{equation}\label{19mar10_conclusion}
 \left[I +\frac{\partial}{\partial P}\left(\tilde{C}_U(P)V_n\right) + \frac{\partial}{\partial P}\left(\tilde{C}_E(P)^tV_n\right) \right]\frac{dP}{dt}
\; = \; (K_2* \tilde{C}_U(P))^t V_n - (L_2* \tilde{C}_E(P))^tV_n.
\end{equation}
The initial value of Eq.~(\ref{19mar10_conclusion}), denoted by $\hat{P}(0)$, must be chosen as the projection  of the initial value $P(0)$ of Eq.~(\ref{nov0809_mainODE}), onto the manifold $\mathcal{M}_0$. The reduction is obtained under the assumption that during the initial transient there has not been any significant change in $\bar{P} = P + C_UV_n + C_E^tV_n$.
Therefore the projection, $\hat{P}(0)$, of the initial conditions onto the slow manifold is related to the original initial conditions, $U(0), P(0), C_U(0), C_E(0),$ by
\begin{equation*}
\hat{P}(0) + \tilde{C}_U(\hat{P}(0))V_n + \tilde{C}_E^t(\hat{P}(0))V_n = P(0) + C_U(0)V_n + C_E(0)V_n = P(0).
\end{equation*}

In summary, if tQSSA is valid,  then Eq.~(\ref{19mar10_conclusion}) is a  reduction of Eq.~(\ref{nov0809_mainODE}). We next study the stability of the slow manifold $\mathcal{M}_0$ defined by Eq.~\eqref{20june10_slowManifold}. This is a necessary step in showing that GSPT can be used to justify the validity of the reduction obtained under the generalized tQSSA.

\subsection{Stability of the slow manifold}\label{20june10_StabilityOfSlowManifold} 


We start by introducing several definitions and some notation to simplify the computations
involved in showing that the slow manifold $\mathcal{M}_0$, defined by Eq.~(\ref{20june10_slowManifold}), is normally hyperbolic and stable. The results also apply to the slow manifolds discussed in sections~\ref{IsolatedMichaelisMentenReaction} and~\ref{sec:may3010_TwoProtein}, as those are particular examples of $\mathcal{M}_0$.


Suppose that $A$ and $B$ are matrices of dimensions $n \times k$ and $n \times l$, respectively.  We denote by $[A \,:\, B]$ the $n \times (k+l)$ matrix obtained by adjoining $B$ to  $A$. 
We use this definition to combine the different coefficient matrices, and let
\[
 C := [C_U\,:\,C_E^t], \quad
 Q_1 := [K_1\,:\,L_1^t], \quad Q_2 := [K_{-1}+K_2\,:\,L_{-1}^t+L_2^t].
\]
We also define
\[
 Z := \left[\begin{array}{c} U \\ E  \end{array} \right], \qquad
 \bar{Z} := \left[\begin{array}{c} U_T - \bar{P} \\ E_T   \end{array} \right], \qquad
I_{2n}^n := \left[\begin{array}{c} I_n  \\  0 \end{array} \right], \qquad \text{and} \qquad
V_{2n} =
\underbrace{ \left[\begin{array}{cccc}1 & 1 & \ldots & 1                \end{array} \right]^t
           }_{2n \text{ times}}. 
\]
Using this notation the right hand side of Eqs.~(\ref{09mar10_closedformb}-\ref{09mar10_closedforma}) can be written  as
\begin{eqnarray*}
 Z = \left[\begin{array}{c} U \\ E   \end{array}\right] 
&=&
\left[\begin{array}{c} U_T - \bar{P} \\ E_T   \end{array} \right] -
\left[\begin{array}{c} C_U^t V_n  \\  C_EV_n  \end{array} \right] +
\left[\begin{array}{c} (I_n*C_U^t) V_n  \\  0 \end{array} \right] \\
 &=&
\bar{Z} - C^tV_n + \left(\left[\begin{array}{c} I_n  \\  0 \end{array} \right]*C^t \right)V_n \\
&=&
\bar{Z} - (C^t - I_{2n}^n*C^t)V_n,
\end{eqnarray*}
and Eq.~(\ref{09mar10_closedformc}) can be written as
$
 P = \bar{P} - CV_{2n}.
$
Therefore, Eqs.~(\ref{09mar10_matrixFormc}-\ref{09mar10_matrixFormb}) can be merged to obtain
\begin{equation}\label{may3110_dcdt}
 \frac{dC}{dt} = \underbrace{Q_1*(PZ^t)-Q_2*C}_{:= F(C)}.
\end{equation}
The manifold $\mathcal{M}_0$ is defined by
\begin{equation*} 
 \mathcal{M}_0 = \left\{C \in \mathbb{R}^{n \times 2n} \,\,   \big| \,\,
Q_1*(PZ^t)-Q_2*C =  F(C) = 0 \right\}.
\end{equation*}

To show that $\mathcal{M}_0$ is stable and normally hyperbolic  we need to show that the Jacobian, 
$ \displaystyle{\frac{\partial F}{\partial C},}$ evaluated at $ {\mathcal{M}_0}$
has eigenvalues with only negative real parts. We will show that $ \displaystyle{\frac{\partial F}{\partial C}}$ has eigenvalues with negative real parts  everywhere, and hence at all points of $ {\mathcal{M}_0},$ \emph{a fortiori}. 

The mapping $F : \mathbb{R}^{n \times 2n} \rightarrow \mathbb{R}^{n \times 2n}$ is a matrix valued function of the matrix variables $C$.  Therefore
$ \displaystyle{\frac{\partial F}{\partial C}}$ represents differentiation with respect to a matrix.
This operation is defined by ``flattening'' a $m \times n$ matrix to a 
$mn \times 1$ vector and taking the gradient.  More precisely, suppose $M = [M_{.1} : M_{.2} : \ldots :  M_{.n} ]$ is a $m \times n$ matrix, where $M_{.j}$ is the $j$th column of $M$. Then define 
\begin{equation}\label{june2110_vecAndHatDefinition}
 \,\mathrm{vec}\,(M) := \left[ \begin{array}{c} M_{.1} \\ M_{.2} \\ \vdots \\ M_{.n} \end{array} \right]  \quad \in \quad\mathbb{C}^{mn \times 1}, \qquad \text{and}  \qquad
\widehat{M} := \, \diag(\,\mathrm{vec}\,(M))  \quad \in \quad \mathbb{C}^{mn \times mn}.
\end{equation}
Therefore, $\,\mathrm{vec}\,(M)$ is obtained  by stacking the columns of $M$ on top of each other, and $\widehat{M}$ is the $mn \times mn$ diagonal matrix whose diagonal entries are given by $\,\mathrm{vec}\,(M)$. 

Suppose $G : \mathbb{C}^{p \times q} \rightarrow \mathbb{C}^{m \times n}$ is a matrix valued function with $ X \in \mathbb{C}^{p \times q} \mapsto G(X) \in \mathbb{C}^{m \times n}$. Then the derivative of $G$ with respect to $X$ is defined as 
\begin{equation}\label{june2110_derivativeDefinition}
 \frac{\partial G}{\partial X} :=  \frac{\partial \,\mathrm{vec}\,(G)}{\partial \,\mathrm{vec}\,(X)},
\end{equation}
where the right hand side is the Jacobian~\citep{neudecker}.  In the appendix we list some important properties of these operators which will be used subsequently (see ~\ref{20june10_diffWrtMatrix}).


A direct application of Theorem~\ref{may3110_HadamardProductRule} stated in ~\ref{20june10_diffWrtMatrix} yields  
\begin{eqnarray*}
\frac{\partial  \,F}{\partial \,C} =
\frac{\partial \,\mathrm{vec}\,(F)}{\partial \,\mathrm{vec}\,(C)} 
&=&
\widehat{Q}_1
\frac{\partial \,\mathrm{vec}\,(PZ^t)}{\partial \,\mathrm{vec}\, (C)}
-\widehat{Q}_2
\frac{\partial \,\mathrm{vec}\,(C)}{\partial \,\mathrm{vec}\,(C)}. 
\end{eqnarray*}

We first assume that all the entries in the connectivity matrices are positive, so that all entries in the matrix $C$ are \emph{actual variables}.  At the end of ~\ref{september1510_stableManifold} we show how to remove this assumption.

Replacing ${\partial \,\mathrm{vec}\,(C)}/{\partial \,\mathrm{vec}\,(C)}$ with the identity matrix,  $I_{2n^2}$, 
  adding $\widehat{Q}_2$ to both side, using Theorems~\ref{may3110_vecABC}, \ref{may3110_hadamardproduct},\ref{may3110_simpleProtuctRule},  \ref{may3110_HadamardProductRule},  and treating $\bar{P}$ and $\bar{Z}$ as independent of $C$ we obtain
\begin{eqnarray*}
\widehat{Q}_2 +\frac{\partial \,\mathrm{vec}\,(F)}{\partial \,\mathrm{vec}\,(C)}
&=&
\widehat{Q}_1 \left[
\left(Z \otimes I_n \right)\frac{\partial \,\mathrm{vec}\,( P)}{\partial \,\mathrm{vec}\, (C) } 
+
\left(I_{2n} \otimes P \right)
\frac{\partial \,\mathrm{vec}\,( Z^t)}{\partial \,\mathrm{vec}\, (C) } 
\right] \\
&=&
\widehat{Q}_1 \left[ 
-
\left(Z \otimes I_n \right)\frac{\partial \,\mathrm{vec}\, (CV_{2n})}{\partial \,\mathrm{vec}\, (C) } 
- 
\left(I_{2n} \otimes P \right)
\frac{\partial \,\mathrm{vec}\,\left(\left(\left(C^t - I_{2n}^n*C^t \right)V_n \right)^t\right)}{\partial \,\mathrm{vec}\, (C) } 
\right] \\
&=&
\widehat{Q}_1 \left[ 
-
\left(Z \otimes I_n \right)\frac{\partial \,\mathrm{vec}\, (CV_{2n})}{\partial \,\mathrm{vec}\, (C) } 
- 
\left(I_{2n} \otimes P \right)
\frac{\partial \,\mathrm{vec}\, \left(V_n^tC - V_n^t(\left({I_{2n}^n}\right)^t*C)\right) }{\partial \,\mathrm{vec}\, (C) } 
\right] \\
&=&
\widehat{Q}_1 \Bigg[ 
-
\left(Z \otimes I_n \right)(V_{2n}^t \otimes I_n)\frac{\partial \,\mathrm{vec}\, (C)}{\partial \,\mathrm{vec}\, (C) } \\
&&-
\left(I_{2n} \otimes P \right) 
\left\{
\frac{\partial \,\mathrm{vec}\, (V_n^tC ) }{\partial \,\mathrm{vec}\, (C) }
-\frac{\partial \,\mathrm{vec}\, ( V_n^t(\left({I_{2n}^n}\right)^t*C)) }{\partial \,\mathrm{vec}\, (C) }
 \right\}
\Bigg] \\
&=&
\widehat{Q}_1 \left[ 
-
\left(ZV_{2n}^t \otimes I_n \right)
-
\left(I_{2n} \otimes P \right) 
\left\{
\left(I_{2n} \otimes V_n^t \right) 
-\left(I_{2n} \otimes V_n^t \right) \widehat{({I_{2n}^n})^t}
 \right\}
\right] \\
&=&
\widehat{Q}_1 \left[ 
-
\left(ZV_{2n}^t \otimes I_n \right)
-
\left(I_{2n} \otimes PV_n^t \right) 
+\left(I_{2n} \otimes PV_n^t \right) \widehat{({I_{2n}^n})^t}
\right] \\
&=&
-\widehat{Q}_1 \left[ 
\left(ZV_{2n}^t \otimes I_n \right)
+
\left(I_{2n} \otimes PV_n^t \right)\left(I_{2n^2}- \widehat{({I_{2n}^n})^t} \right)
\right]. 
\end{eqnarray*}
Here
$
 \widehat{({I_{2n}^n})^t}  
$ is the matrix obtained by applying the operator defined in Eq.~(\ref{june2110_vecAndHatDefinition}) to the transpose of $I_{2n}^n$. 

This computation shows that the Jacobian matrix of interest has the form
\begin{equation}\label{1april10_jacobian}
 J := 
\frac{\partial  \,F}{\partial \,C}
=
-\widehat{Q}_1 \left[ 
\left(ZV_{2n}^t \otimes I_n \right)
+
\left(I_{2n} \otimes PV_n^t \right)\left(I_{2n^2}- \widehat{\left({I_{2n}^n}\right)^t} \right)
\right] - \widehat{Q}_2.
\end{equation}
The following Lemma, proved in the~\ref{september1510_stableManifold}, shows that this Jacobian matrix always has eigenvalues with negative real part.

\begin{lemma}\label{15mar10_stableMatrix}
Suppose $Z \in \mathbb{R}^{2n \times 1}_+$ is a $2n$ dimensional vector with positive entries,  $Y \in \mathbb{R}^{n \times 1}_+$ is an $n$ dimensional vector with positive entries, $\Lambda, \Gamma \in \mathbb{R}^{2n^2 \times 2n^2}$ are diagonal matrices with positive entries on the diagonal. Further assume that $R_{n} $ and $ R_{2n}$ are row vectors of size $n$ and $2n$ respectively with all entries equal to $1$. Then the $2n^2 \times 2n^2$ matrix
\begin{equation}\label{15mar10_jacobian}
J = \Lambda \left[(ZR_{2n} \otimes I_n) + (I_{2n}\otimes YR_n)\left(I_{2n^2}- \widehat{\left({I_{2n}^n}\right)^t} \right) \right] + \Gamma 
\end{equation}
has eigenvalues with strictly positive real parts.
\end{lemma}
This Lemma applies to connectivity matrices with strictly positive entries. In \ref{Section:Sparse} we show how
to generalize the Lemma to the case when the connectivity matrices contain zero entries. In this case only the principal submatrix of the Jacobian, $J$, corresponding to the positive entries of the connectivity matrices needs to be examined. Since any principal submatrix of $J$ inherits the stability properties of $J$, the result follows.   We therefore obtain the following corollary.

\begin{corollary}
The manifold $\mathcal{M}_0$ defined in Eq.~\eqref{20june10_slowManifold} is normally hyperbolic and stable.
\end{corollary}

\subsection{Validity of tQSSA in the general setup}\label{ValidityOfTQSSAInGeneralSetup}

We next investigate the asymptotic limits under which the tQSSA is valid in the general setting described at the beginning of this section.  We follow the approach given in the previous sections
to obtain a suitable rescaling of the variables.  While this rescaling does not change the 
stability of the slow manifold, $\mathcal{M}_0$,   it allows us to more easily describe the 
asymptotic limits in which the timescales are separated, and the system is singularly perturbed.

Recall that  Eq.~(\ref{09mar10_matrixForm}) and Eq.~(\ref{may3010_ODEtermWise}) are equivalent. The concise form given in Eq.~(\ref{09mar10_matrixForm}) was useful in obtaining a reduction and checking the stability of the  slow manifold. However, to obtain sufficient conditions for the validity of the tQSSA, we will work with  Eqs.~(\ref{may3010_ODEtermWise}) and (\ref{21mar_xieiyi}).  

Let $l_{ij}^m := (l_{ij}^{-1}+l_{ij}^2)/l_{ij}^1$, $k_{ij}^m := (k_{ij}^{-1}+k_{ij}^2)/k_{ij}^1$ denote the MM constants.  Then the following scaling factors are natural generalizations of those introduced in section~\ref{sec:may3010_TwoProtein}, 
\begin{equation*}
\beta_{ij} := \frac{E_i^TU_j^T}{E_i^T+U_j^T+l_{ij}^m}, 
\quad
\alpha_{ij} := \frac{U_i^TU_j^T}{U_i^T+U_j^T+k_{ij}^m}, 
\quad \quad \quad
i,j \in \{1,2,...,n\}.
\end{equation*}
 Note that for each pair $(i,j)$ either all of $k_{ij}^1, k_{ij}^{-1}, k_{ij}^2$ are all zero or all nonzero.  In the case that $k_{ij}^1 = k_{ij}^{-1} = k_{ij}^2 = 0$  we define $k_{ij}^m := 0$. Similarly, if $l_{ij}^1 = l_{ij}^{-1} = l_{ij}^2 = 0$ then $l_{ij}^m := 0$. Let
\begin{equation*}
T_{\bar{U}} := \, \max \left\{\, \underset{i,j}{\max}\left\{ \frac{U_j^T}{l_{ij}^2 \beta_{ij}}\right\},\, \underset{i,j}{\max}\left\{ \frac{U_j^T}{k_{ij}^2 \alpha_{ij}}\right\} \right\}  
= \frac{U_{j_0}^T}{l_{i_0j_0}^2 \beta_{i_0j_0}}, \quad \text{for some } i_0,j_0 \in \{1,2,...,n\}.
\end{equation*}
We next define the following dimensionless rescaling of the variables in 
Eq.~(\ref{may3010_ODEtermWise}) 
\begin{equation}\label{may3110_rescaling}
 \tau = \frac{t}{T_{\bar{U}}}, \quad \text{and} \quad 
\bar{p}_i(\tau) = \frac{\bar{P}_i(t)}{U_i^T}, \quad 
c_{ij}^u(\tau)= \frac{C_{ij}^U(t)}{\alpha_{ij}}, \quad 
c_{ij}^e(\tau) = \frac{C_{ij}^E(t)}{\beta_{ij}}, \quad i,j \in \{1,2,...,n\}.
\end{equation}
After rescaling,  Eqs.~(\ref{may3010_ODEtermWise}) take the form 
\begin{subequations}\label{21mar10_rescaled}
\begin{equation}
\frac{d\bar{p}_i}{d\tau} 
= \sum_{r=1}^n\bigg(   \frac{k_{ri}^{2}\alpha_{ri}U_{j_0}^T}{l_{{i_0j_0}}^2 \beta_{{i_0j_0}} U_i^T}c_{rj}^u - \frac{l_{ri}^2 \beta_{ri} U_{j_0}^T}{l_{{i_0j_0}}^2 \beta_{{i_0j_0}} U_i^T} c_{rj}^e\bigg), 
\label{21mar10_rescaleda} 
\end{equation} 
%
%
\addtocounter{equation}{1}
\begin{eqnarray}
\bigg(\frac{ \beta_{ij}  }{l_{ij}^1  E_i^TU_j^T T_{\bar{U}}}\bigg)
\frac{dc_{ij}^e}{d\tau} 
&=& 1-c_{ij}^e  -\left[ \sum_{\underset{s\neq j}{s=1}}^n\frac{\beta_{is}}{E_i^T}c_{is}^e \right]
    \left[ 1 -\bar{x}_j-\frac{1}{U_j^T}\sum_{\underset{r\neq i}{r=1}}^n\left(\beta_{rj}c_{rj}^e+\sum_{\underset{s\neq j}{s=1}}^n\alpha_{js}c_{js}^u\right)  \right]  \nonumber   \\
&&-\frac{1}{U_j^T} 
    \left[  U_j^T\bar{x}_j+\sum_{\underset{r\neq i}{r=1}}^n\beta_{rj}c_{rj}^e+\sum_{\underset{s\neq j}{s=1}}^n\alpha_{js}c_{js}^u  \right]\nonumber \\
&& -\frac{1}{U_j^T}\left[  U_j^T\bar{x}_j +\sum_{\underset{r\neq i}{r=1}}\beta_{rj}c_{rj}^e+\sum_{\underset{s\neq j}{s=1}}^n\alpha_{js}c_{js}^u + \sum_{\underset{s\neq i}{s=1}}^n\beta_{is}c_{is}^e \right]\frac{\beta_{ij}c_{ij}^e}{E_i^T}  +  \frac{(\beta_{ij}c_{ij}^e)^2}{E_i^TU_j^T}.
\nonumber \\
 \label{21mar10_rescaledb}
\end{eqnarray}
\end{subequations} 
The rescaled form of Eq.~(\ref{may3010_ODEtermWisec}) is similar to the rescaled form of  Eq.~(\ref{may3010_ODEtermWiseb}), and we therefore omit it.
If we define 
\[
 \epsilon_{ij} := \frac{k^2_{ij}}{k^1_{ij}}\frac{U^T_i}{(U^T_i + U^T_j + k^m_{ij})^2}, 
\quad
 \epsilon^e_{ij} := \frac{l^2_{ij}}{l^1_{ij}}\frac{E^T_i}{(E^T_i + U^T_j + l^m_{ij})^2}, 
\]
and let
\begin{equation}\label{june2110_epsilon}
 \epsilon := \text{max} \left\{ \underset{i,j}{\text{max}} \left\{ \epsilon_{ij} \right\}, \underset{i,j}{\text{max}} \left\{  \epsilon^e_{ij}\right\} \right\},
\end{equation}
then  the following theorem defines the conditions under which  Eq.~(\ref{21mar10_rescaled}) defines a singularly perturbed system and, hence, conditions under which  GSPT is applicable. 

\begin{theorem}\label{20june10_TQSSA}
If for all non-zero $k_{ij}^1,k_{ij}^2,k_{ij}^{-1}$ and for all non zero $l_{ij}^1,l_{ij}^2,l_{ij}^{-1}$ and for all $U_i^T,E_i^T$ 
\[
\begin{array}{ccccccc}
 \mathcal{O}\left(\frac{k^{1}_{ij}}{k^{1}_{rs}}\right)
 &=& 
\mathcal{O}\left(\frac{l^{1}_{ij}}{k^{1}_{rs}}\right)
 &=&
\mathcal{O}\left(\frac{l^{1}_{ij}}{l^{1}_{rs}}\right) 
&=& \mathcal{O}(1), \\
\mathcal{O}\left(\frac{k^{2}_{ij}}{k^{2}_{rs}}\right)
 &=& 
\mathcal{O}\left(\frac{l^{2}_{ij}}{k^{2}_{rs}}\right)
 &=&
\mathcal{O}\left(\frac{l^{2}_{ij}}{l^{2}_{rs}}\right) 
&=& \mathcal{O}(1), \\
\mathcal{O}\left(\frac{k^{-1}_{ij}}{k^{-1}_{rs}}\right)
 &=& 
\mathcal{O}\left(\frac{l^{-1}_{ij}}{k^{-1}_{rs}}\right)
 &=&
\mathcal{O}\left(\frac{l^{-1}_{ij}}{l^{-1}_{rs}}\right) 
&=& \mathcal{O}(1), \\
\mathcal{O}\left(\frac{X_i^T}{X_j^T}\right)
 &=& 
\mathcal{O}\left(\frac{X_i^T}{E_j^T}\right)
 &=& 
\mathcal{O}\left(\frac{E_i^T}{E_j^T}\right) 
&=& \mathcal{O}(1),
\end{array}
\qquad 
1 \le i,j,r,s \le n,
\]
in the limit $\epsilon \rightarrow 0$, then Eq.~(\ref{21mar10_rescaled}) is a singularly perturbed system with the structure of Eq.~(\ref{E:group}).  In particular, the $\bar{p}_i$ are the slow variables,  and the $c_{ij}$ and $c_{ij}^e$ are the fast variables. 
\end{theorem}
\begin{proof}
For each $i$ there always exist indices $r,s$ such that $k_{ri}^2 \neq 0 \neq k_{si}^2 $. Hence, the the right hand side of Eq.~(\ref{21mar10_rescaleda}) is not identically zero for any $i \in \{1,2,...,n\}$. Furthermore, by assumption all coefficients on the right hand side of Eq.~(\ref{21mar10_rescaleda}) are $\mathcal{O}(1)$ as $\epsilon \rightarrow 0$. This implies that $\epsilon$ times the right hand side of Eq.~(\ref{21mar10_rescaleda}) is identically zero, in the limit $\epsilon \rightarrow 0$.

Secondly, the definition of $\beta_{ij}$ implies that all coefficients on the right hand side of Eq.~(\ref{21mar10_rescaledb}) are less than or equal to 1. Also, by definition, at least one coefficient has value exactly equal to 1. Hence, the right hand side of Eq.~(\ref{21mar10_rescaledb}) is not identically zero in the limit $\epsilon \rightarrow 0$. 

The definitions of $\epsilon, \alpha_{ij}, \beta_{ij}, T_{\bar{U}}$ imply that coefficients of $\frac{dc^e_{ij}}{d\tau}$ in Eq.~(\ref{21mar10_rescaledb}) are less than or equal to $\epsilon$. For example
\[
\frac{ \beta_{ij}  }{l_{ij}^1  E_i^TU_j^T} \frac{1}{ T_{\bar{U}}}
\le
\frac{ \beta_{ij}  }{l_{ij}^1  E_i^TU_j^T} \frac{l^2_{ij} \beta_{ij}}{U_j^T }
=
\epsilon_{ij}^e
\le
\epsilon.
\]
Hence, in the limit $\epsilon \rightarrow 0$, the left hand side of Eq.~(\ref{21mar10_rescaledb}) vanishes while the right hand side does not. To conclude the proof we only need to show the stability of the slow manifold in rescaled coordinates. But we have already shown that for unscaled coordinates in section~\ref{20june10_StabilityOfSlowManifold} and a non-singular scaling of variable, as in Eq.~(\ref{may3110_rescaling}), will not affect the eigenvalues of the Jacobian.
\end{proof}
Hence, under the assumptions of the above theorem, Eq.~(\ref{21mar10_rescaled}) has the form of Eq.~\ref{10may10_GSPT}. Hence, switching back to unscaled variables  we conclude that in the limit $\epsilon \rightarrow 0$, tQSSA is valid, \emph{i.e.} the reduction from Eq.~(\ref{09mar10_matrixForm}) to Eq.~(\ref{09mar10_matrixFormDEL}) is valid.

\subsection{The assumption of zero initial concentrations of intermediate complexes and the choice of scaling}\label{Sect.ZeroInitComplex} 

Before concluding, we discuss the significance of zero initial concentrations of intermediate complexes and the benefit of the choice of scaling we used to verify the asymptotic limits in which the system is singularly perturbed. Proposition~\ref{june2110_invariantHypercube} below proves that if the reaction  starts with zero initial concentration of intermediate complexes then the solution of both Eqs.~(\ref{09mar10_matrixForm}) and (\ref{21mar10_rescaled}) are trapped in an  $\mathcal{O}(1)$ neighborhood of the origin. Hence, separation of time scale in Eq.(\ref{21mar10_rescaled}), implied by Theorem~\ref{20june10_TQSSA} can be used to obtain the reduction of Eq.~(\ref{09mar10_matrixForm}) given by Eq.~(\ref{09mar10_matrixFormDEL}).
This is important, since GSPT would not be applicable if the rescaling were to send $\mathcal{O}(1)$ solutions of Eq.~(\ref{09mar10_matrixForm}) to  solutions  of Eq.~(\ref{21mar10_rescaled}) that are unbounded 
as $\epsilon \rightarrow 0$.

\begin{proposition}\label{june2110_invariantHypercube}
The $2n^2+n$ dimensional hypercube $\Omega$ defined by
\[
 \Omega := \left\{\{\bar{p}_i\},\{c_{ij}^u\},\{c_{ij}^e\}\, |\, 
 0\le\bar{p}_i \le 1,\,
 0 \le c_{ij}^u\le 2,\,
 0 \le c_{ij}^e\le 2,\,
 \forall \,  i,j \in \{1,2,...,n\}\right\},
\]
is invariant under the flow of Eq.~(\ref{21mar10_rescaled}).
\end{proposition}
\begin{proof}
By the  construction of the differential equations from the law of mass action, all the species concentration variables can take only non negative values. This together with 
the conservation constraints
(\ref{20mar10_UT})  force the $\bar{P_i}$ to take values between $0$ and $U_i^T$. Therefore $0\le\bar{p}_i(\tau) \le 1,\, \forall \, \tau >0$, provided the initial conditions are chosen in $\Omega$.


Positivity of variables also implies that $c_{ij}^u(\tau) \ge 0,\,c_{ij}^e(\tau) \ge 0$ if the flow starts inside $\Omega$. So we only need to show that $c_{ij}^u(\tau) \le 2$ and $ \,c_{ij}^e(\tau) \le 2$. It is sufficient to show that 
$
 \frac{d c_{ij}^u} {d \tau} \bigg |_{c_{ij}^u = 2} \le 0,$ and
 $
 \frac{d c_{ij}^e} {d \tau} \bigg |_{c_{ij}^e = 2} \le 0,$
or equivalently that $ \frac{d C_{ij}^U} {d t } \bigg |_{C_{ij}^U = 2\alpha_{ij}} \le 0,$ and $
 \frac{d C_{ij}^E} {d t } \bigg |_{C_{ij}^E = 2\beta_{ij}} \le 0.
$
But 
\begin{eqnarray*}
  \frac{d C_{ij}^U} {d t } \bigg |_{C_{ij}^U = 2\alpha_{ij}} 
&=&
  k_{ij}^1\left[P_i U_j  - (k_{ij}^{-1}+k_{ij}^{2}) C_{ij}^U \right]\big |_{C_{ij}^U = 2\alpha_{ij}} \\
&=&
k_{ij}^1\Bigg[\left(\bar{P}_i - \sum_{s=1}^nC_{is}^U - \sum_{r=1}^nC_{ri}^E \right)  \left( U_i^T - \bar{P}_i -   \sum_{r=1}^nC_{ri}^U\right)  \\
&& \hspace{7cm}
  - (k_{ij}^{-1}+k_{ij}^{2}) C_{ij}^U \Bigg]\Bigg |_{C_{ij}^U = 2\alpha_{ij}} \\
&\le& 
  k_{ij}^1\left(P_i^T-2\alpha_{ij} \right) \left(U_j - 2\alpha_{ij}\right) - (k_{ij}^{-1}+k_{ij}^{2}) 2\alpha_{ij}  \\ 
&=&
  k_{ij}^1 \left [\left(P_i^T-2\alpha_{ij} \right) \left(U_j^T - 2\alpha_{ij}\right) - k_{ij}^{m} 2\alpha_{ij} \right] \\
&\le& 0. 
\end{eqnarray*}
Similarly we can show that $C_{ij}^E$ is decreasing when $C_{ij}^E = \beta_{ij}$. This concludes the proof.
 \end{proof}

 From this we conclude that the assumptions of Theorem~\ref{20june10_TQSSA} and the zero initial values of intermediate complexes together imply the tQSSA.

Finally, we combine the results of section~\ref{20june10_StabilityOfSlowManifold} with Theorem~\ref{20june10_TQSSA} and Proposition~\ref{june2110_invariantHypercube} to obtain the main result of this study.

\begin{theorem}\label{june2110_mainTheorem}
If the parameters of  Eq.~(\ref{nov0809_mainODE}) are such that assumptions of Theorem~\ref{20june10_TQSSA} are satisfied and the initial values of intermediate complexes are zero, then the tQSSA holds.  For $\epsilon$ defined by Eq.~(\ref{june2110_epsilon}), there exists an $\epsilon_0$ such that for all $0 < \epsilon< \epsilon_0$, the solutions of  Eq.~(\ref{09mar10_matrixForm}) are
$\mathcal{O}(\epsilon)$ close to the solutions of Eq.~(\ref{09mar10_matrixFormDEL}) after an exponentially fast transient.  Eq.~(\ref{nov0809_mainODE}) can therefore be reduced to the $n$ dimensional Eq.~(\ref{19mar10_conclusion}) involving only the protein concentrations, $P_i$.
\end{theorem}

\section{Discussion}\label{sec:discussion}

We obtained sufficient condition for the validity of tQSSA in non-isolated Michaelis-Menten type reactions.  We therefore significantly generalized  previous approaches that extended  the MM scheme to small networks of 
reactions~\citep{PedersenBersaniBersaniCortese2008}, and provided a theoretical justification of the numerical results obtained in~\citep{CilibertoFabrizioTyson2007}. 

 We noted that the  direct application of the tQSSA to equations modeling networks of reactions produces a reduction that contains coupled quadratic equations. However, for the class of networks discussed here  we were able to circumvent this problem by solving and equivalent linear system.  Moreover,  we obtained a closed form equation in terms of protein concentrations only.  A direct application of the tQSSA leads to a reduced system that involves the concentration of proteins and intermediate complexes.  It was also shown that the slow manifold used in the system reduction is always attracting.

MM type reactions are often used in  models of signaling networks.  In such models it is frequently assumed that the reduced equation describing the dynamics of a single, isolated protein can be used to study interactions in networks.  It has been noted that this use of MM differential equations is not necessarily justified~\citep{CilibertoFabrizioTyson2007}. The present approach provide an alternative approximation that was proved to be valid. 

Recently, a general reduction procedure for multiple timescale chemical reaction networks has been proposed~\citep{othmerLee2010}. That study considered a general chemical interaction network, with a pre--determined set of fast and slow reactions.  We deal with a more restrictive class of equations, which makes it unnecessary to start with a prior knowledge of fast and slow reactions. Moreover, we are able to  show the normal hyperbolicity of the slow manifold
in our reduction, something that was not possible in the more general setting described in~\citep{othmerLee2010}.


We  end by pointing out a couple of limitations of this work. Firstly,
not all enzymatic networks belong to the class we have considered here. For example, our full reduction scheme does not work for the network depicted in Fig. \ref{22april10_notReducable}. 

\clearpage
\begin{figure}[t]
\begin{center}
\includegraphics[scale=.2]{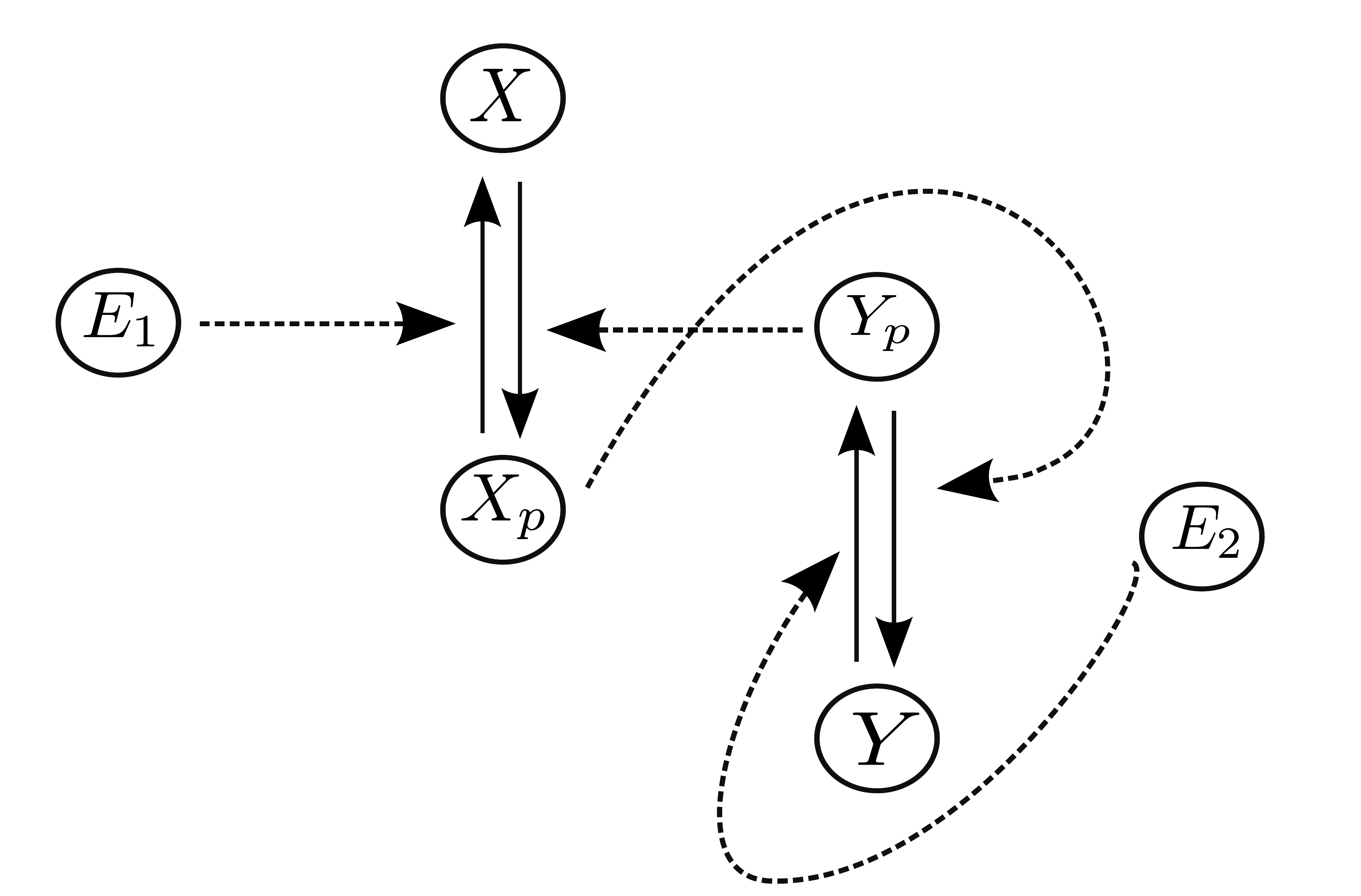}
\end{center}
\caption{\footnotesize{A hypothetical network for which the reduction described in sections~\ref{S:general_coordinates}--\ref{S:general_reduced} leads to a differential--algebraic system of equations.  The concentrations of the intermediate complexes appear in a nonlinear way
in the resulting algebraic equations.  A further reduction to a form involving only the protein concentrations
is therefore not apparent. } }
\label{22april10_notReducable}
\end{figure}
\clearpage

This network is a slight modification of the network in  Fig.~\ref{25feb10_fig1}a). Although the tQSSA can be justified, the algebraic part of the reduced equations cannot be solved using our approach. These equations have  the form
\begin{eqnarray*}
 0 &=& \underbrace{(X_T - X_p - C_x^e - C_x - C_y)}_{= X}   
 \underbrace{(Y_T - Y - C_y - C_x - C_y^e )}_{ = Y_p}                 - k_mC_x, \\
 0 &=& X_p \underbrace{(Y_T - Y - C_y - C_x - C_y^e )}_{ = Y_p}              - k_mC_y, \\
 0 &=& (E_1^T -C_x^e )X_p                - k_mC_x^e, \\
 0 &=& (E_2^T -C_y^e ) Y                - k_mC_y^e, 
\end{eqnarray*}
which has to be solved for $C_x ,C_y,C_x^e,C_y^e $ in terms of $X_p,Y $. Immediately we run into problems because the first equation in the above algebraic system is quadratic in the unknown variables. 

We also note that  no approximation theory is truly complete unless error bounds are investigated. Although GSPT guarantees that the derived approximations are  $\mathcal{O}(\epsilon)$ close to 
the true solutions,  a more precise description of the error terms may be desired.

{\bf Acknowledgements:} We thank Patrick de Leenheer, Paul Smolen and Antonios Zagaris
for helpful discussions and comments on earlier version of the manuscript.   This work was supported by NSF Grants 
DMS-0604429 and DMS-0817649 and a Texas ARP/ATP award.

\appendix
\section{Bound on the expression for $\epsilon$ as obtained in Eq.~(\ref{27mar10_eps})}

\begin{lemma}\label{lemma1}
(Bound on $\epsilon$): If $k_1,k_2,k_{-1},e,x \,\in \mathbb{R}_+$, then  
\[
 \epsilon := \frac{k_2}{k_1}\frac{e}{(e+x+\frac{k_{-1}+k_2}{k_1})^2} \le \frac{ k_1e\,k_2}{(k_1e+k_2)^2} \le \frac{1}{4}.
\]
\end{lemma}
\begin{proof}
Since $k_1,k_2,k_{-1},e,x$ are all positive,
\begin{eqnarray*}
 \frac{k_2}{k_1}\frac{e}{(e+x+\frac{k_{-1}+k_2}{k_1})^2} 
&\le& \frac{k_2}{k_1}\frac{e}{(e+\frac{k_2}{k_1})^2} 
= \frac{ k_1e\,k_2}{(k_1e+k_2)^2}.
\end{eqnarray*}
Since for any positive number $s$, $s+1/s \ge 2$, we obtain
\begin{eqnarray*}
\frac{ k_1e\,k_2}{(k_1e+k_2)^2} \le \frac{ 1}{\left(\sqrt{\frac{k_1e}{k_2}}+\sqrt{\frac{k_2}{k_1e}}\,\right)^2 } \le \frac{1}{4}.
\end{eqnarray*}
\end{proof}
This bound is sharp because for $k_1 = 1$, $k_2 = 1 $, $k_{-1} \rightarrow 0$, $e =1, x \rightarrow 0$ we obtain $\epsilon \rightarrow 1/4$.

\section{Differentiation with respect to a matrix}\label{20june10_diffWrtMatrix}
The theory of differentiation with respect to a matrix  is described in~\citep{neudecker}. We already introduced the  $\mathrm{ vec }$ and  \emph{hat}  operators and the definitions of differentiation with respect to a matrix variable in Eqs.~(\ref{june2110_vecAndHatDefinition}) and (\ref{june2110_derivativeDefinition}).  Below we list some important properties of these operators as they relate to differentiation with respect to a matrix. Proofs can be found in~\citep{neudecker}.
\begin{theorem}[\citep{roth1934,Neudecker1969Some}]\label{may3110_vecABC}
For any three matrices $A,B$ and $C$ such that the matrix product $ABC$ is defined,
\[
 \mathrm{ vec }\, (ABC) = (C^t \otimes A) \,\mathrm{vec}\, (B).
\]
\end{theorem}
\begin{theorem}[\citep{neudecker}]\label{may3110_hadamardproduct}
For any two matrices $A$ and $B$ of equal size
\[
 \,\mathrm{vec}\, (A * B) = \widehat{A} \,\mathrm{vec}\,(B) = \widehat{B} \,\mathrm{vec}\,(A).
\]
\end{theorem}

\begin{theorem}[\textbf{Product rule}\citep{neudecker}]\label{may3110_simpleProtuctRule}
 Let $G : \mathbb{C}^{p \times q} \rightarrow \mathbb{C}^{m \times r}$ and $H : \mathbb{C}^{p \times q} \rightarrow \mathbb{C}^{r \times n}$ be two differentiable function then
\[
 \frac{\partial \,\mathrm{vec}\,(GH)}{\partial \,\mathrm{vec}\,(X)}
=
(H^t \otimes I_m)\frac{\partial \,\mathrm{vec}\,(G)}{\partial \,\mathrm{vec}\,(X)}
+
(I_n \otimes G)\frac{\partial \,\mathrm{vec}\,(H)}{\partial \,\mathrm{vec}\,(X)}.
\]
\end{theorem}

\begin{theorem}[{\bf Hadamard product rule} \citep{neudecker}]\label{may3110_HadamardProductRule} 
 Let $G : \mathbb{C}^{p \times q} \rightarrow \mathbb{C}^{m \times n}$ and $H : \mathbb{C}^{p \times q} \rightarrow \mathbb{C}^{m \times n}$ be two differentiable functions then
\[
 \frac{\partial \,\mathrm{vec}\,(G*H)}{\partial \,\mathrm{vec}\,(X)}
=
\widehat{H}\, \frac{\partial \,\mathrm{vec}\,(G)}{\partial \,\mathrm{vec}\,(X)}
+
\widehat{G}\,\frac{\partial \,\mathrm{vec}\,(H)}{\partial \,\mathrm{vec}\,(X)}.
\]
\end{theorem}

\section{Proof of Lemma~  \ref{21june10_linerSolve}}\label{june0110_solveLinearEquation}
Note that the unknowns in Eq.~(\ref{10mar10_linear}) are matrices and the structure of the equation is somewhat similar to a Lyapunov equation, $AX + XB = C$. A standard approach to solving Lyapunov equations  is to vectorize the matrices (see~\citep{matrixAnalysisII}), resulting in an equation of the type $\left[(I_m \otimes A) + (B^t\otimes I_n ) \right] \,\mathrm{vec}\, (X) = \,\mathrm{vec}\, (C)$. Proving solvability then essentially reduces to proving the non-singularity of the coefficient matrix $\left[(I_m \otimes A) + (B^t\otimes I_n ) \right]$.  We will use this approach to show the solvability of Eq.~(\ref{10mar10_linear}).

In the proof of this Lemma we first assume that all possible reactions occur at nonzero rates so that all entries in the  matrices $K_1,K_2,K_{-1},L_1,L_2$ and $L_{-1}$ are strictly positive.  The result is
then generalized to the case when some reaction rates are zero, so that no all reactions  occur.

Note that Eq.~(\ref{10mar10_linearb}) is uncoupled from Eq.~(\ref{10mar10_lineara}). Using Theorems~\ref{may3110_vecABC} and \ref{may3110_hadamardproduct} from section~\ref{20june10_diffWrtMatrix}, we vectorize Eq.~(\ref{10mar10_linearb})  to obtain
\begin{eqnarray}
 \,\mathrm{vec}\, \big[ L_1* \left(C_E \left(V_n  P^t\right)\right)&+&  (L_{-1}+L_2)* C_E \big]  \nonumber\\
&=& \,\mathrm{vec}\, \left[ L_1* \left(C_E \left(V_n  P^t\right)\right)\right] + \,\mathrm{vec}\, \left[ (L_{-1}+L_2)* C_E \right] \nonumber\\
&=&  \widehat{L}_1 \,\mathrm{vec}\, \left[C_E \left(V_n  P^t\right)\right] + (\widehat{L}_{-1}+\widehat{L}_2)  \,\mathrm{vec}\,(C_E)  \nonumber\\
&=&  \widehat{L}_1 \left(PV^t_n \otimes I_n\right) \,\mathrm{vec}\,(C_E)+ (\widehat{L}_{-1}+\widehat{L}_2)  \,\mathrm{vec}\,(C_E)  \nonumber\\
&=&  \left[ \widehat{L}_1 \left(PV^t_n \otimes I_n\right)+ (\widehat{L}_{-1}+\widehat{L}_2) \right] \,\mathrm{vec}\,(C_E) \label{june0110_vecCE}
\end{eqnarray}
The following lemma shows that the matrix multiplying $\,\mathrm{vec}\,(C_E)$ in this equation is invertible. 
\begin{lemma}\label{invertCoeff}
If $A,B \in \mathbb{R}_+^{n^2 \times n^2}$ are diagonal matrices with positive entries on the diagonal, $Y \in \mathbb{R}_+^{n\times 1}$  is a column vector with positive entries , $V_n = [1\,1\,\cdots\,1]^t$ is a column vector of size $n$, and $I_n$ is the $n\times n$ identity matrix, then the $n^2 \times n^2$ matrix 
\[
D =  A \left(YV_n^t \otimes I_n\right)+ B
\]
is invertible.
\end{lemma}
\begin{proof}
Invertibility of $D$ is equivalent to invertibility of $B^{-1}D$.  Therefore it is sufficient to prove the result with $
 B = I_{n^2 \times n^2} =: I$, so that  
$
D =  A \left(YV_n^t \otimes I_n\right)+ I.
$
If  $A \left(YV^t_n \otimes I_n\right)$ does not have $-1$ as an eigenvalue, then  $D$ cannot have $0$ as an eigenvalue. Demonstrating this will complete the proof. 
Let 
\begin{equation*}
A = \left[\begin{array}{cccc}
     A_1 &   &  &   \\
         & A_2   &  &   \\
         &    & \ddots  &   \\ 
   &    &    & A_n   \\   
    \end{array}\right], \quad
Y = \left[\begin{array}{c}
     y_1 \\
  y_2 \\
\vdots \\
y_n
    \end{array}\right],
\end{equation*}
where $A_i \in \mathbb{R}_+^{n\times n}$, $i \in \{1,2,...,n\}$ are diagonal matrices, and $y_i \in \mathbb{R}_+$. 
Now 
\begin{eqnarray*}
YV^t_n \otimes I_n &=& \left[\begin{array}{c}
     y_1 \\
  y_2 \\
\vdots \\
y_n
    \end{array}
\begin{array}{c}
     y_1 \\
  y_2 \\
\vdots \\
y_n
    \end{array}...
\begin{array}{c}
     y_1 \\
  y_2 \\
\vdots \\
y_n
    \end{array}
\right]\otimes I_n 
= \left[\begin{array}{c}
     y_1I_n \\
  y_2I_n \\
\vdots \\
y_nI_n
    \end{array}
\begin{array}{c}
     y_1I_n \\
  y_2I_n \\
\vdots \\
y_nI_n
    \end{array}
...
\begin{array}{c}
     y_1I_n \\
  y_2I_n \\
\vdots \\
y_nI_n
    \end{array}
\right].
\end{eqnarray*}
This implies that 
\begin{eqnarray}\label{11mar10_break}
 A \left(YV^t_n \otimes I_n\right) = 
\left[\begin{array}{c}
     y_1A_1 \\
  y_2A_2 \\
\vdots \\
y_nA_n
    \end{array}\begin{array}{c}
     y_1A_1 \\
  y_2A_2 \\
\vdots \\
y_nA_n
    \end{array}
...
\begin{array}{c}
     y_1A_1 \\
  y_2A_2 \\
\vdots \\
y_nA_n
    \end{array}\right].
\end{eqnarray}
Suppose $\lambda$ is an eigenvalue of $A \left(YV_n^t \otimes I_n\right)$, and
\[
 \bar{X} = \left[\begin{array}{c}
     X_1 \\
  X_2 \\
\vdots \\
X_n
    \end{array}\right],
\]

$X_i \in \mathbb{C}^{n\times 1}, \, i \in \{1,2,...,n \} $  the corresponding eigenvector. Using Eq.~(\ref{11mar10_break}) we have
\begin{eqnarray*}
 \left[\begin{array}{c}
     y_1A_1 \\
  y_2A_2 \\
\vdots \\
y_nA_n
    \end{array}\begin{array}{c}
     y_1A_1 \\
  y_2A_2 \\
\vdots \\
y_nA_n
    \end{array}
...
\begin{array}{c}
     y_1A_1 \\
  y_2A_2 \\
\vdots \\
y_nA_n
    \end{array}\right]
\left[\begin{array}{c}
     X_1 \\
  X_2 \\
\vdots \\
X_n
    \end{array}\right]
 =
 \lambda 
\left[\begin{array}{c}
     X_1 \\
  X_2 \\
\vdots \\
X_n
    \end{array}\right].
\end{eqnarray*}
This implies that for all $k \in \{1,2,...,n\}$,
\begin{eqnarray}\label{11mar10_bingo}
 \left[\begin{array}{c}
     y_1A_{1_k} \\
  y_2A_{2_k} \\
\vdots \\
y_nA_{n_k}
    \end{array}\begin{array}{c}
     y_1A_{1_k} \\
  y_2A_{2_k} \\
\vdots \\
y_nA_{n_k}
    \end{array}...\begin{array}{c}
     y_1A_{1_k} \\
  y_2A_{2_k} \\
\vdots \\
y_nA_{n_k}
    \end{array}\right]
\left[\begin{array}{c}
     X_{1_k} \\
  X_{2_k} \\
\vdots \\
X_{n_k}
    \end{array}\right] = \lambda \left[\begin{array}{c}
     X_{1_k} \\
  X_{2_k} \\
\vdots \\
X_{n_k}
    \end{array}\right],
\end{eqnarray}
where $A_{i_k}$ is $(k,k)$-th entry in the matrix $A_i$, and $X_{i_k}$ is the $k$th entry in the vector $X_i$.

Therefore, if $\lambda$ is an eigenvalue of $A \left(YV^t_n \otimes I_n\right)$ then it must be an eigenvalue of one of its $n\times n$ principal submatrices which have the form of the coefficient matrix in Eq.~(\ref{11mar10_bingo}) and whose eigenvalues we know are either zero or $\sum_{i=1}^ny_iA_{i_k}$ (see reason in the footnote\footnote{
We have
\begin{eqnarray*}
 \left[\begin{array}{c}
     y_1A_{1_k} \\
  y_2A_{2_k} \\
\vdots \\
y_nA_{n_k}
    \end{array}\begin{array}{c}
     y_1A_{1_k} \\
  y_2A_{2_k} \\
\vdots \\
y_nA_{n_k}
    \end{array}...\begin{array}{c}
     y_1A_{1_k} \\
  y_2A_{2_k} \\
\vdots \\
y_nA_{n_k}
    \end{array}\right]^t
\left[\begin{array}{c}
     1 \\
  1 \\
\vdots \\
1
    \end{array}\right] = \sum_{i=1}^ny_iA_{i_k} \left[\begin{array}{c}
     1 \\
    1 \\
   \vdots \\
   1
    \end{array}\right].
\end{eqnarray*}
Since the coefficient matrix in the above equation is rank one,  $\sum_{i=1}^ny_iA_{i_k}$ is the only non-zero eigenvalue.
}).
 Hence $\lambda$ can not be $-1$, and hence $D$ cannot have a zero eigenvalue.  \end{proof}
This settles the problem of solvablity of $C_E$ in Eq.~(\ref{10mar10_linearb}). We can use this solution to eliminate  $C_E$ from Eq.~(\ref{10mar10_lineara}). Rewriting Eq.~(\ref{10mar10_lineara}) with all the known terms on the right hand side we obtain
\begin{eqnarray}\label{june0110_CU}
K_1* \big[P \big(  V_n^tC_U^t + V_n^tC_U &-& V_n^t (I_n * C_U)\big)  \big] 
+
(K_{-1}+K_2)* C_U \nonumber \\
&=&\quad K_1* \left[P \left( U_T^t- P^t\right)\right] 
-K_1* \left[ PV_n^t C_E  \right]. 
\end{eqnarray}
We can write 
\begin{eqnarray}\label{july2810_vecBig}
 \,\mathrm{vec}\,\left[P \left(  V_n^tC_U^t + V_n^tC_U- V_n^t (I_n * C_U)\right)  \right]
&=&
(I_n \otimes P) \, 
 \,\mathrm{vec}\,\left[  V_n^tC_U^t + V_n^tC_U- V_n^t (I_n * C_U)  \right]. \nonumber \\ 
\end{eqnarray}
Since $(C_UV_n)^t$ is a row vector, we have $\,\mathrm{vec}\,[(C_UV_n)^t] =\,\mathrm{vec}\,(C_UV_n) $.  Therefore, using Theorems~\ref{may3110_vecABC} and \ref{may3110_hadamardproduct}
\begin{eqnarray*}
 \,\mathrm{vec}\,(V_n^tC_U^t) 
&=&
 \,\mathrm{vec}\,(C_UV_n) = (V_n^t \otimes I_n) \,\mathrm{vec}\,(C_U), 
\\
 \,\mathrm{vec}\,(V_n^tC_U)
&=&
( I_n \otimes V_n^t )  \,\mathrm{vec}\,(C_U),
\\
\,\mathrm{vec}\,\left(V_n^t (I_n * C_U)\right)
&=&( I_n \otimes V_n^t )\,\mathrm{vec}\, (I_n * C_U) = ( I_n \otimes V_n^t ) \widehat{I}_n \,\mathrm{vec}\,(C_U).
\end{eqnarray*}
Plugging these in Eq.~(\ref{july2810_vecBig}) we get
\begin{eqnarray*}
 \,\mathrm{vec}\,\big[P \big(  V_n^tC_U^t + V_n^tC_U 
&-&
 V_n^t (I_n * C_U)\big)  \big] \\
&=&
(I_n \otimes P) 
\left[
(V_n^t \otimes I_n)
+
( I_n \otimes V_n^t )
-
( I_n \otimes V_n^t ) \widehat{I}_n
 \right] \,\mathrm{vec}\,(C_U)
\\
&=&
\left[
(I_n \otimes P) (V_n^t \otimes I_n)
+
( I_n \otimes PV_n^t )
-
( I_n \otimes PV_n^t ) \widehat{I}_n
 \right]\,\mathrm{vec}\,(C_U).
\end{eqnarray*}
The vectorized form of the left hand side of Eq.~(\ref{june0110_CU}) is 
\begin{equation*}
\left[
 \widehat{K}_1
\left\{
(I_n \otimes P) (V_n^t \otimes I_n)
+
( I_n \otimes PV_n^t )
-
( I_n \otimes PV_n^t ) \widehat{I}_n
 \right \} + 
(\widehat{K}_{-1} +\widehat{K}_{-1} )
\right]
\,\mathrm{vec}\,(C_U).
\end{equation*}
The following Lemma shows that the  matrix mutliplying $\,\mathrm{vec}\,( C_U)$ in this expression is invertible.

\begin{lemma}\label{17may10_invertCoefficient}
If $A,B \in \mathbb{R}_+^{n^2 \times n^2}$ are diagonal matrices with positive entries on the diagonal, $Y \in \mathbb{R}_+^{n\times 1}$  is a column vector with positive entries, $V_n = [1\,1\,\cdots\,1]^t$ is a column vector of size $n$, then the $n^2 \times n^2$ matrix 
\[
D =  A \left(
\left(I_n \otimes Y \right)
\left(V_n^t \otimes I_n \right)
+
\left(I_n \otimes YV_n^t \right)
-
\left(I_n \otimes YV_n^t \right)\widehat{I}_n
\right)+ B
\]
is invertible.
\end{lemma}
\begin{proof}
The invertibility of $D$ is equivalent to invertibility of $A^{-1}D$. We can therefore assume that $A = I_{n^2}$.
Now 
\begin{eqnarray*}
 \left(I_n \otimes Y \right)
\left(V_n^t \otimes I_n \right)
&=&
\left[
\everymath{}
\footnotesize
\begin{array}{cccc|cccc|c|cccc} 
\Y & \Z & \ddots & \Z & \Y & \Z & \ddots & \Z & \ldots & \Y & \Z & \ddots & \Z \vspace{1mm} \\
\hline \\
\Z & \Y & \ddots & \Z & \Z & \Y & \ddots & \Z & \ldots & \Z & \Y & \ddots & \Z \vspace{1mm}\\ 
\hline \\
\, & \, &     \, & \, & \, & \, &     \, & \, & \ddots &\, & \, &     \, & \, \\
\, & \, &     \, & \, & \, & \, &     \, & \, & \ddots &\, & \, &     \, & \, \vspace{1mm} \\ 
\hline \\
\Z & \Z & \ddots & \Y & \Z & \Z & \ddots & \Y & \ldots &\Z & \Z & \ddots & \Y
\end{array}
\right], 
\end{eqnarray*}
and
\begin{eqnarray*}
\left(I_n \otimes YV_n^t \right) 
= \left[
\begin{array}{cccc}
     \YVn^t &        &        & \\
            & \YVn^t &        & \\
            &        & \ddots & \\
            &        &        & \YVn^t 
    \end{array} \right].
\end{eqnarray*}
So,
\begin{eqnarray*}
\left(I_n \otimes YV_n^t \right)(I_{n^2} - \widehat{I}_n) = 
 \left[
\everymath{}
\scriptsize
\begin{array}{cccc}
     \YVan^t &        &        & \\
            & \YVbn^t &        & \\
            &        & \ddots & \\
            &        &        & \YVnn^t 
    \end{array} \right],
\end{eqnarray*}
and
\[
 \left(I_n \otimes Y \right)
\left(V_n^t \otimes I_n \right)
+
\left(I_n \otimes YV_n^t \right)
-
\left(I_n \otimes YV_n^t \right)\widehat{I}_n 
\]
\begin{equation*}
=
 \left[
\everymath{}
\footnotesize
\begin{array}{c|c|c|c}
     \YVn^t & \Za    &  \ldots & \Za \vspace{1mm} \\
\hline \\
     \Zb    & \YVn^t &  \ldots & \Zb  \vspace{1mm}   \\
\hline \\
            &        & \ddots  &     \\
           &        & \,  &    \vspace{1mm} \\
\hline \\
      \Zn      &  \Zn      &  \ldots       & \YVn^t 
    \end{array} \right].
\end{equation*}

Clearly, its sufficient to show the invertibility of $D$ with $y_1 = y_2 = ... = y_n = 1$.
We examine
\begin{eqnarray*}
D - B 
= \left[
\everymath{}
\footnotesize
\begin{array}{c|c|c|c}
     \OYVn^t & \OZa    &  \ldots & \OZa \vspace{1mm}\\
\hline \\
     \OZb    & \OYVn^t &  \ldots & \OZb   \vspace{1mm}  \\
\hline \\
            &        & \ddots  &     \\
           &        & \,  &   \vspace{1mm}  \\
\hline \\
      \OZn      &  \OZn      &  \ldots       & \OYVn^t 
    \end{array} \right].
\end{eqnarray*}
Now let
\[
 V = \left[\begin{array}{cccccccccccccccccccc}
\begin{array}{ccc} v_{11}  & \ldots & v_{1n} \end{array}&
\begin{array}{ccc} v_{21}  & \ldots & v_{2n} \end{array}&
\ldots &
\begin{array}{ccc} v_{n1}  & \ldots & v_{nn} \end{array}
 \end{array}\right]^t
\]
be an eigenvector of $D$ corresponding to a zero eigenvalue. We aim to show that $V = 0$.
Let 

\[
 B = \text{ diag}\, \left[
\begin{array}{cccccccccc}
 b_{11} & \cdots  & b_{1n}  &
 b_{21} & \cdots  & b_{2n} &
  \cdots                               &
 b_{n1} &  \cdots & b_{nn} 
\end{array}
\right].
\]
Then  for each $i,j \in \{1,2,...,n \}$,
\begin{equation}\label{sept1410_main}
 \underbrace{\sum_{s=1}^nv_{is} +  \sum_{\underset{r \neq i}{r=1}}^nv_{ri}}_{:= -\lambda_i} = -b_{ij}v_{ij}.
\end{equation}
Note that the left hand side of this equation, which we denote by $-\lambda_i$,  is independent of $j$. Hence, for all $i,j \in \{1,2,...,n \}$ we obtain $v_{ij} = \frac{\lambda_i}{b_{ij}}$. Using this observation in  Eq.~(\ref{sept1410_main}) we get 

\[
 \sum_{s=1}^n\frac{\lambda_i}{b_{is}} +  \sum_{\underset{r \neq i}{r=1}}^n\frac{\lambda_r}{b_{ri}} = -\lambda_i, \qquad \forall \,  i \in \{1,2,...,n \}.
\]
This equality can be written in matrix form as
\[
 \left[
\begin{array}{rrrrrrrrrrrrrrrr}
1 + \sum_{s=1}^n\frac{1}{b_{1s}} & \frac{1}{b_{21}} & \quad\ldots &  \frac{1}{b_{n1}} \vspace{1mm} \\
\frac{1}{b_{12}} & 1 + \sum_{s=1}^n\frac{1}{b_{2s}}  & \quad\ldots &  \frac{1}{b_{n2}} \\
 & & \ddots & 
\\
\frac{1}{b_{1n}} &\frac{1}{b_{2n}} &  \quad\ldots & 1 + \sum_{s=1}^n\frac{1}{b_{ns}}
\end{array}
\right]
\left[
\begin{array}{c}
 \lambda_1 \\  \lambda_2 \\ \vdots \\ \lambda_n
\end{array}
\right] 
= 0
\]
The coefficient matrix is diagonally dominant along the columns, and hence invertible. This implies that $\lambda_i = 0$, and so $v_{ij} = 0$. 
 \end{proof}

Lemmas~\ref{invertCoeff} and \ref{17may10_invertCoefficient} together complete the proof of Lemma~\ref{21june10_linerSolve} for the case when all the entries in the connectivity matrices are strictly positive. This proof can be extended to  general connectivity matrices, as stated in the Lemma~\ref{21june10_linerSolve} in the following way. 

Suppose that some of  the entries in the connectivity matrix are zero.
Let,
\begin{subequations}\label{E:sept1710_A}
\begin{alignat}{2}
I_K &= [I_K(i,j)]_{i,j = 1}^n \quad \text{such that}\quad I_K(i,j) &\,=&\begin{cases}
                                                                     1 & \text{if } k_{ij}^1, k_{ij}^{-1}, k_{ij}^2 \text{ are nonzero,} \\
                                                                     0 & \text{if } k_{ij}^1 = k_{ij}^{-1} =  k_{ij}^2 = 0,
                                                                   \end{cases}
\\
I_L &= [I_L(i,j)]_{i,j = 1}^n \quad \text{such that}\quad I_L(i,j) &\,=& \begin{cases}
                                                                     1 & \text{if } l_{ij}^1, l_{ij}^{-1}, l_{ij}^2 \text{ are nonzero,} \\
                                                                     0 & \text{if } l_{ij}^1 = l_{ij}^{-1} =  l_{ij}^2 = 0.
                                                                   \end{cases}
\end{alignat}
\end{subequations}
Hence $I_K$ and $I_L$ are the unweighted connectivity matrices of the reaction network.
The matrices of intermediate variables, corresponding to existing connections, now have the form
\begin{equation}\label{newCUCE}
 C_U^{I_K} = I_K *  C_U \qquad  C_E^{I_L} = I_L *  C_L.
\end{equation}

Replacing $C_U$ with $C_U^{I_K}$ and  $C_E$ with $C_E^{I_E}$ in Eq.~(\ref{10mar10_linear}), one can easily check that the solution of the non-zero entries of $C_U^{I_K}$ and   $C_E^{I_E}$ does not depend on the zero entries of $K_1,$ $K_2,$ $K_{-1},$ $L_1,$ $L_2,$ $L_{-1}$. This observation completes the proof of Lemma~\ref{21june10_linerSolve}.

\section{Stability and normal hyperbolicity of the  slow manifold ${\mathcal M}_0$}\label{september1510_stableManifold}

Following the approach in the previous section, we first prove the result under the assumption that    $K_1,K_2,K_{-1},L_1,L_2$ and $L_{-1}$ are strictly positive.  At the end of this section we show how to generalize the proof to the case when some of the
reactions do not occur.

First we start with a  preliminary lemma.
\begin{lemma}\label{5may10_stableMatrix}
Suppose $Z \in \mathbb{R}^{2n \times 1}_+$ is a $2n$ dimensional vector with positive entries,  $Y \in \mathbb{R}^{n \times 1}_+$ is an $n$ dimensional vector with positive entries, and $\widehat{\Psi} = [\widehat{\psi}_{ij}],\, \widehat{\Gamma} = [\widehat{\gamma}_{ij}] \in \mathbb{R}^{n \times 2n}_+$ real matrices with positive entries. Let $\lambda \in \mathbb{C}$ be a complex number with \emph{nonpositive real part}.  If $V = [v_{ij}] \in \mathbb{C}^{n \times 2n} $ is a complex matrix that satisfies the following system of linear homogeneous equations,
\begin{subequations}\label{16mar10_AB}
\begin{eqnarray}
\frac{1}{y_i}\sum_{s=1}^{2n}{v_{is}
+
\frac{1}{z_j}\sum_{\underset{r \neq j}{r=1}}^n{v_{rj}}}
&=&
\frac{\widehat{\psi}_{ij}}{y_iz_j}\left( \lambda-\widehat{\gamma}_{ij}\right) v_{ij}, \quad  \begin{array}{r} 1 \le i \le n, \\  1 \le j \le n,  \end{array}
\\
\frac{1}{y_i}\sum_{s=1}^{2n}{v_{is}
+
\frac{1}{z_j}\sum_{r=1}^n{v_{rj}}}
&=&
\frac{\widehat{\psi}_{ij}}{y_iz_j}\left( \lambda-\widehat{\gamma}_{ij}\right) v_{ij}, \quad
\begin{array}{r} 1 \le i \le n, \\ n+1 \le j \le 2n, \end{array}
\end{eqnarray}
\end{subequations}
then $V$ is the zero matrix.
\end{lemma}

\begin{proof} 
 Let $V = [v_{ij}] \in \mathbb{C}^{n \times 2n} $ satisfy Eq.~(\ref{16mar10_AB}). We will show that $v_{ij} = 0 $ for all $i,j$. 
Let
\begin{eqnarray*}
 R_i := \sum_{j=1}^{2n}v_{ij}, \quad  1 \le i \le n, \qquad
 C_j :=  
\begin{cases}
\sum_{\underset{i \neq j}{i=1}}^{n}v_{ij}, \quad  1 \le j \le n, \\
 \sum_{i=1}^{n}v_{ij}, \quad  n+1 \le j \le 2n.
\end{cases}
\end{eqnarray*}
Then Eq.~(\ref{16mar10_AB}) can be written as
\[
 \frac{1}{y_i}R_i
+
\frac{1}{z_j}C_j
=\frac{\widehat{\psi}_{ij}}{y_iz_j}\left( \lambda-\widehat{\gamma}_{ij}\right) v_{ij}, 
\quad
\begin{array}{r} 1 \le i \le n, \\  1 \le j \le 2n,  \end{array}
\]
Setting $a_{ij} = \frac{\widehat{\psi}_{ij}}{y_iz_j}\left( \lambda-\widehat{\gamma}_{ij}\right)$, we have
\begin{equation}\label{19may10_A}
 \frac{1}{a_{ij}y_i}R_i
+
 \frac{1}{a_{ij}z_j}C_j
= v_{ij}, 
\quad 1 \le i \le n, \, 1 \le j \le 2n.
\end{equation}
By summing Eq.~(\ref{19may10_A}) over $i$ and $j$ separately we obtain the following system of linear equations in the unknowns $\{R_1,R_2,...,R_n,$ $ C_1,C_2,...,C_{2n}\}$ 
\begin{subequations}\label{18may10_B}
 \begin{align}
 R_i \frac{1}{y_i}\sum_{j=1}^{2n} \frac{1}{a_{ij}}
+
\sum_{j=1}^{2n} \frac{1}{z_j a_{ij}} C_j
&= R_i,\quad 1 \le i \le n, \label{18may10_Ba} \\
\sum_{i=1}^n \frac{1}{y_ia_{ij}}R_i
+
C_j \frac{1}{z_j} \sum_{i=1}^n\frac{1}{a_{ij}}
&= C_j, \quad  1 \le j \le 2n \label{18may10_Bb}
\end{align} 
\end{subequations}
Eq.~(\ref{18may10_B}) can be written in matrix form  as
\begin{equation}\label{18may10_C}
\underbrace{
 \left[
\everymath{}
\footnotesize
\begin{array}{ccc|ccc}
-1+\frac{1}{y_1}\sum_{j=1}^{2n} \frac{1}{a_{1j}} &  &  &  \frac{1}{z_1a_{11}} & \ldots & \frac{1}{z_{2n}a_{1,{2n}}} \\
                                             & \ddots &   & \vdots & & \vdots\\
 & & -1+\frac{1}{y_n}\sum_{j=1}^{2n} \frac{1}{a_{nj}} & \frac{1}{z_1a_{n1}} & \ldots & \frac{1}{z_{2n}a_{n,2n}}   \\ &&&&& \\
\hline &&&&& \\
\frac{1}{y_1 a_{11}} & \ldots & \frac{1}{y_n a_{n1}} &
-1 + \frac{1}{z_1}\sum_{i=1}^n \frac{1}{a_{i1}} &  &  
 \\  
  \vdots  &  & \vdots  &  & \ddots &  \\
\frac{1}{y_1 a_{1,2n}} & \ldots & \frac{1}{y_n a_{n{2n}}} &
 &  &
-1 + \frac{1}{z_{2n}}\sum_{i=1}^n \frac{1}{a_{i, {2n}}}
\end{array}
\right]
}_{:= A}
 \left[
\begin{array}{c}
R_1 \\ R_2 \\ \vdots \\ R_n\\ C_1 \\ C_2 \\ \vdots \\ C_{2n}
\end{array}
\right]
=
0.
\end{equation}

We next show that the  the coefficient matrix, $A$, is invertible. This will imply that $R_i = C_j = 0, \, \forall \, i,j$. This, together with (\ref{19may10_A}), will force $v_{ij}$ to be zero and we will be done. 

To show the non-singularity of  $A$ it is sufficient to show the non singularity of the product of   $A$ with  a non-singular diagonal matrix
\[
 A \left[
\begin{array}{ccccccccccccccccccc}
 y_1 & & &     \\
     & \ddots &  \\
     &        & y_n \\
     &        &     &  z_1 & & &     \\
     &        &     &      &   \ddots &  \\
     &        &     &      &          & z_{2n} \\
\end{array}
\right]
\]
\[
 = 
\underbrace{
 \left[
\begin{array}{ccc|ccc}
-y_1+\sum_{j=1}^{2n} \frac{1}{a_{1j}} &  &  &  \frac{1}{a_{11}} & \ldots & \frac{1}{a_{1,{2n}}} \\
                                             & \ddots &   & \vdots & & \vdots\\
 & & -y_n+\sum_{j=1}^{2n} \frac{1}{a_{nj}} & \frac{1}{a_{n1}} & \ldots & \frac{1}{a_{n,2n}}   \\ &&&&& \\
\hline &&&&& \\
\frac{1}{a_{11}} & \ldots & \frac{1}{ a_{n1}} & -z_1 + \sum_{i=1}^n \frac{1}{a_{i1}} &  &  \\  
  \vdots  &  & \vdots  &  & \ddots &  \\
\frac{1}{ a_{1,2n}} & \ldots & \frac{1}{a_{n{2n}}} &  &  & -z_{2n} +\sum_{i=1}^n \frac{1}{a_{i, {2n}}}
\end{array}
\right].}_{= X}
\]

\noindent
Note that $X$ is a complex symmetric matrix (i.e. $X = X^t$). To show the non singularity of $X$, it is sufficient to show that $X$ has no zero eigenvalue. Assume that $\alpha$ is an eigenvalue of $X$ and $u \in \mathbb{R}^{3n}$  a corresponding eigenvector. Break $X$ into two Hermitian matrices,
\[
 X = \underbrace{\frac{X + X^*}{2}}_{:= S} + i \underbrace{\frac{X - X^*}{2i}}_{:= T} = S +i T,
\]
where $X^*$ is the conjugate transpose of $X$). Then, 
\begin{eqnarray*}
 \alpha \langle u,u \rangle = \langle Xu,u \rangle = \langle Su,u \rangle + i \langle Tu,u \rangle.
\end{eqnarray*}
To show that $\alpha $ is not zero, it is sufficient to show that $\langle Su,u \rangle$ is not zero for any  $0 \neq u \in \mathbb{R}^{3n}$. Note that, since $S,$ and $T$ are Hermitian, the terms $\langle Su,u \rangle$ and $\langle Tu,u \rangle$) are always real. 

But since $X$ is a complex symmetric matrix, $S_{ij} = \frac{X_{ij}+\bar{X}_{ji}}{2} =\frac{X_{ij}+\bar{X}_{ij}}{2} = Re (X_{ij}) $, where $S_{ij},$ and  $X_{ij}$ are the $(i,j)$-th entries of the matrices $S $ and $X$ respectively, and $\bar{X}_{ij}$  is the complex conjugate of the complex number $X_{ij}$, and $Re (X_{ij})$ is the real part of  $X_{ij}$.
Therefore, 
\[
S = 
\left[
\everymath{}
\footnotesize
\begin{array}{ccc|ccc}
-y_1+\sum_{j=1}^{2n} Re\,\frac{1}{a_{1j}} &  &  & Re\,\frac{1}{a_{11}} & \ldots & Re\,\frac{1}{a_{1,{2n}}} \\
                                             & \ddots &   & \vdots & & \vdots                         							\\
 & & -y_n+\sum_{j=1}^{2n} Re\,\frac{1}{a_{nj}} & Re\,\frac{1}{a_{n1}} & \ldots & Re\,\frac{1}{a_{n,2n}}  	\\ 
&&&&& 																				\\
\hline 
&&&&&																				\\
Re\,\frac{1}{ a_{11}} & \ldots & Re\,\frac{1}{ a_{n1}} &
-z_1 + \sum_{i=1}^n Re\,\frac{1}{a_{i1}} &  &  												\\  
  \vdots  &  & \vdots  &  & \ddots &  															\\
Re\,\frac{1}{ a_{1,2n}} & \ldots & Re\,\frac{1}{ a_{n,2n}} &
 &  &
-z_{2n} + \sum_{i=1}^n Re\, \frac{1}{a_{i, {2n}}}
\end{array}
\right] . 
\]
Recall that $a_{ij} = \frac{\widehat{\psi}_{ij}}{y_iz_j}\left( \lambda-\widehat{\gamma}_{ij}\right) $. If the real part of $\lambda$ is nonpositive then the real parts of $a_{ij}$ are negative. This implies that $Re\,\frac{1}{a_{ij}} < 0$ for all  $i,j$. In turn,  this implies that $S$ is diagonally dominant, and all the eigenvalues of $S$ are negative and real, since $S$ is a real symmetric matrix. 

Therefore $\langle Su,u \rangle < 0$ for all $u \in \mathbb{R}^{3n}$, and 
 $\alpha$ cannot be zero. This implies that $X$ is invertible, which further implies that $A$ is invertible. So, $R_i = C_j = 0$ for $i,j$. Eq.~(\ref{19may10_A})  therefore implies that $V = 0$.  
  \end{proof}


\subsection{Proof of Lemma~\ref{15mar10_stableMatrix}:}
\begin{proof}
 We will prove the lemma by contradiction. 
 Let 
\begin{equation*}
 Z = \left[ \begin{array}{c} z_1 \\ z_2 \\ \vdots \\ z_{2n} \end{array} \right],
\quad
Y = \left[ \begin{array}{c} y_1 \\ y_2 \\ \vdots \\ y_{n} \end{array} \right].
\end{equation*}
Then
\begin{eqnarray*}
 ZR_{2n} \otimes I_n = \left[ \begin{array}{c} z_1 \\ z_2 \\ \vdots \\ z_{2n} \end{array} 
\begin{array}{c} z_1 \\ z_2 \\ \vdots \\ z_{2n} \end{array} 
\hdots
\begin{array}{c} z_1 \\ z_2 \\ \vdots \\ z_{2n} \end{array} 
\right]\otimes I_n
= \underbrace{
\left[ \begin{array}{c} z_1I_n \\ z_2I_n \\ \vdots \\ z_{2n}I_n \end{array} 
\begin{array}{c} z_1I_n \\ z_2I_n \\ \vdots \\ z_{2n}I_n \end{array} 
\hdots
\begin{array}{c} z_1I_n \\ z_2I_n \\ \vdots \\ z_{2n}I_n \end{array} \right]}_{2n \text{ block columns}}, 
\end{eqnarray*}

\begin{eqnarray*}
 I_{2n}\otimes YR_n  = I_{2n}\otimes \underbrace{
\left[ \begin{array}{c} y_1 \\ y_2 \\ \vdots \\ y_{n} \end{array} 
\begin{array}{c} y_1 \\ y_2 \\ \vdots \\ y_{n} \end{array} 
\hdots
\begin{array}{c} y_1 \\ y_2 \\ \vdots \\ y_{n} \end{array} \right]}_{=YR_n}
=\underbrace{
\left[
\begin{array}{cccc} 
YR_n &       &       & \\
     &  YR_n &       & \\
     &       &\ddots & \\
     &       &       & Y R_n 
\end{array}
\right]}_{2n \text{ block columns}}.
\end{eqnarray*}

Let 
$
R_n^{\left(i\right)} = \left[\begin{array}{ccccccccc} 
                              1 &\cdots& 1 & 0 & 1 \cdots & 1
                             \end{array}\right]
$ be a row vector with a zero in the $i$-th place and 1s everywhere else. 
Then, 
\begin{eqnarray*}
 (I_{2n}\otimes YR_n )\left(I_{2n^2}- \widehat{{I_{2n}^n}^t} \right)
  =
\underbrace{
\left[
\begin{array}{ccccccccc} 
YR_n^{\left(1\right)} &        &                        &       &        &     \\
                      & \ddots &                        &       &        &     \\
                      &        & YR_n^{\left(n\right)}  &       &        &     \\
                      &        &                        &  YR_n &        &     \\
                      &        &                        &       & \ddots &     \\
                      &        &                        &       &        & YR_n
\end{array}
\right]}_{2n \text{ block columns}}.
\end{eqnarray*}

Therefore,
\[
 ZR_{2n} \otimes I_n + (I_{2n}\otimes YR_n )\left(I_{2n^2}- \widehat{{I_{2n}^n}^t} \right) \]
\begin{eqnarray*}
=
\left[ \begin{array}{lll|lll}
z_1I_n+YR_n^{\left(1\right)} &  \ldots    &  z_1I_n                      &    z_1I_n        & \ldots         & z_1I_n       \\
\vdots                       &  \ddots    & \vdots                       &   \vdots         & \ddots         & \vdots       \\
z_nI_n                       &  \ldots    & z_nI_n+YR_n^{\left(n\right)} &   z_nI_n         & \ldots         & z_nI_n       \\
\hline
 z_{n+1}I_n                  &  \ldots    &  z_{n+1}I_n                  & z_{n+1}I_n+YR_n  & \ldots         &  z_{n+1}I_n  \\
 \vdots                      &  \ddots    & \vdots                       &  \vdots          & \ddots         &  \vdots      \\ 
 z_{2n}I_n                   &  \ldots    & z_{2n}I_n                    &  z_{2n}I_n       & \ldots         & z_{2n}I_n+YR_n
\end{array} \right].
\end{eqnarray*}

Let 
\begin{equation*}
\Lambda = \left[\begin{array}{cccc}
\Lambda^{\left(1\right)} & & & \\
 &\Lambda^{\left(2\right)} & & \\
  & & \ddots & \\
 & & & \Lambda^{\left(2n\right)} 
\end{array}\right],
\quad
\Gamma = \left[\begin{array}{cccc}
\Gamma^{\left(1\right)} & & & \\
 &\Gamma^{\left(2\right)} & & \\
  & & \ddots & \\
 & & & \Gamma^{\left(2n\right)} 
\end{array}\right],
\end{equation*}
where $\Lambda^{\left(k\right)},\Gamma^{\left(k\right)}$, $k \in \{1,2,...,2n\}$ are $n \times n$ diagonal blocks of $\Lambda, \Gamma$ respectively.
Hence
\[
 J = \left[ \begin{array}{c|c} A_{11} &  A_{12} \\ \hline A_{21} &  A_{22}   \end{array}  \right],
\]
where
\begin{eqnarray*}
A_{11} 
&=& 
\left[ \begin{array}{lll}
z_1\Lambda^{\left(1\right)}+\Lambda^{\left(1\right)}YR_n^{\left(1\right)} +\Gamma^{\left(1\right)}&  \ldots    &  z_1\Lambda^{\left(1\right)}                            \\
\vdots                       &  \ddots    & \vdots     \\
z_n\Lambda^{\left(n\right)}  &  \ldots    & z_n\Lambda^{\left(n\right)}+\Lambda^{\left(n\right)}YR_n^{\left(n\right)}+\Gamma^{\left(n\right)} 
\end{array} \right],
\\
A_{12} 
&=& 
\left[ \begin{array}{lll}
   z_1\Lambda^{\left(1\right)}       & \ldots         & z_1\Lambda^{\left(1\right)}       \\
                      \vdots         & \ddots         & \vdots                            \\
   z_n\Lambda^{\left(n\right)}       & \ldots         & z_n\Lambda^{\left(n\right)}       
\end{array} \right],
\\
A_{21} 
&=& 
\left[ \begin{array}{lll}
 z_{n+1}\Lambda^{\left(n+1\right)}    &  \ldots    &  z_{n+1}\Lambda^{\left(n+1\right)} \\
 \vdots                               &  \ddots    & \vdots                             \\ 
 z_{2n}\Lambda^{\left(2n\right)}      &  \ldots    & z_{2n}\Lambda^{\left(2n\right)}   
\end{array} \right],
\\
A_{22} 
&=& 
\left[ \begin{array}{lll}
 z_{n+1}\Lambda^{\left(n+1\right)} +\Lambda^{\left(n+1\right)} YR_n+\Gamma^{\left(n+1\right)}  & \ldots         &  z_{n+1}\Lambda^{\left(n+1\right)}   \\
  \vdots          & \ddots         &  \vdots      \\ 
  z_{2n}\Lambda^{\left(2n\right)}        & \ldots         & z_{2n}\Lambda^{\left(2n\right)} +\Lambda^{\left(2n\right)} YR_n+\Gamma^{\left(n+1\right)}
\end{array} \right].
\end{eqnarray*}

Let $\lambda$ be an eigenvalue of $J$, with a corresponding eigenvector 
\[
V = \left[\begin{array}{c} V^{\left(1\right)} \\V^{\left(2\right)}\\ \vdots \\V^{\left(2n\right) } \end{array}\right] \in \mathbb{C}^{2n^2}, \quad \text{where }  
V^{\left(k\right)}
= 
\left[
\begin{array}{c}
 v_1^{\left(k\right)} \\
 v_2^{\left(k\right)} \\
 \vdots \\
 v_n^{\left(k\right)}
\end{array}
\right]
 \in \mathbb{C}^{n}, k \in \{1,2,...,2n\}.
\]
We will show that $v_l^{\left(k\right)} = 0$ for all $l,k$. 
By definition of eigenvalues and using the block structure of $J $ we get
\begin{eqnarray*}
\left[\begin{array}{c}
z_1\Lambda^{\left(1\right)}\sum_{j=1}^{2n}{V^{\left(j\right)}}+\Lambda^{\left(1\right)}YR_n^{\left( 1\right)}V^{\left(1\right)} + \Gamma^{\left(1\right)}V^{\left(1\right)} \\
\vdots
\\
z_n\Lambda^{\left(n\right)}\sum_{j=1}^{2n}{V^{\left(j\right)}}+\Lambda^{\left(n\right)}YR_n^{\left( n\right)}V^{\left(n\right)} +\Gamma^{\left(n\right)}V^{\left(n\right)} 
\\ 
z_{n+1}\Lambda^{\left(n+1\right)}\sum_{j=1}^{2n}{V^{\left(j\right)}}+\Lambda^{\left(n+1\right)}YR_nV^{\left(n+1\right)} +\Gamma^{\left(n+1\right)}V^{\left(n+1\right)} 
\\
\vdots
\\
z_{2n}\Lambda^{\left(2n\right)}\sum_{j=1}^{2n}{V^{\left(j\right)}}+\Lambda^{\left(2n\right)}YR_nV^{\left(2n\right)} +\Gamma^{\left(2n\right)}V^{\left(2n\right)}
\end{array}\right]
=
\left[\begin{array}{l} \lambda V^{\left(1\right)} 
\\ \vdots \\
\lambda V^{\left(n\right)}
\\
\lambda V^{\left(n+1\right)}
\\
\vdots
\\
\lambda V^{\left(2n\right) }
\end{array}\right].
\end{eqnarray*}
Looking at the above equation row by row we get
\begin{subequations}\label{20may10_row}
\begin{eqnarray}
z_k\Lambda^{\left(k\right)}\sum_{j=1}^{2n}{V^{\left(j\right)}}+\Lambda^{\left(k\right)}YR_n^{\left(k\right)}V^{\left(k\right)} + \Gamma^{\left(k\right)}V^{\left(k\right)} 
&=&
\lambda V^{\left(k\right)}, \qquad k \in \{1,2,...,n\} \label{20may10_rowa} \\
z_k\Lambda^{\left(k\right)}\sum_{j=1}^{2n}{V^{\left(j\right)}}+\Lambda^{\left(k\right)}YR_nV^{\left(k\right)} + \Gamma^{\left(k\right)}V^{\left(k\right)} 
&=&
\lambda V^{\left(k\right)}, \qquad k \in \{n+1,...,2n\} \label{20may10_rowb}
\end{eqnarray}
\end{subequations}
Note that Eq.~(\ref{20may10_row}) is still in matrix multiplication form. Writing it further in terms of each of its rows,  for each $k \in \{1,2,...,2n\}$ and $l \in \{1,2,...,n\}$, we have (For notational simplicity let $\left(\Lambda^{\left(k\right)}\right)^{-1} := \Psi^{\left(k\right)}$ )
\begin{subequations}\label{16mar10_A}
\begin{eqnarray}
\frac{1}{y_l}\sum_{j=1}^{2n}{v_l^{\left(j\right)}}
+
\frac{1}{z_k}\sum_{\underset{h \neq l}{h=1}}^n{v_h^{\left(k\right)}}
&=&\frac{\psi_l^{\left(k\right)}}{y_lz_k}\left( \lambda-\gamma_l^{\left(k\right)}\right) v_l^{\left(k\right)},
\qquad k \in \{1,...,n\},
\\
\frac{1}{y_l}\sum_{j=1}^{2n}{v_l^{\left(j\right)}}
+
\frac{1}{z_k}\sum_{h=1}^n{v_h^{\left(k\right)}}
&=&\frac{\psi_l^{\left(k\right)}}{y_lz_k}\left( \lambda-\gamma_l^{\left(k\right)}\right) v_l^{\left(k\right)},
\qquad k \in \{n+1,...,2n\}.
\end{eqnarray}
\end{subequations}

Now, Lemma~\ref{5may10_stableMatrix} applied to Eq.~(\ref{16mar10_A}) immediately yields that $v_l^{\left(k\right)} = 0$ for all $l,k$. This implies that the real part of $\lambda$ cannot be nonpositive. This completes the proof of stability of $J$.
\end{proof}

\subsection{Stability of slow manifold in the absence of some connections}\label{Section:Sparse} 
Lemma~\ref{15mar10_stableMatrix} show that the slow manifold defined by Eq.~\eqref{20june10_slowManifold}
is normally hyperbolic and stable when all entries in the connectivity matrices are positive. We 
next show how to extend the result to the case when some reactions are absent.

Recall the definitions of the unweighted connectivity matrices, $I_K, I_L$, and the associated matrices $C_U^{I_K}$ and $C_E^{I_L}$ given in  Eqs.~(\ref{E:sept1710_A}) and (\ref{newCUCE}). Let $I_{L^t}$  be a $n \times n$ matrix with ones at the places where $L_1^t , L_2^t, L_{-1}^t$ are non zero and zero where $L_1^t , L_2^t, L_{-1}^t$ are zeros. Now, recall the definition of $C$ in section~\ref{20june10_StabilityOfSlowManifold} and define 
\[
 C_0 = \left[\begin{array}{cc} I_U  & I_{L^t} \end{array}\right] * C
\]
Then, in the sense that we only need to differentiate along the coordinates corresponding to positive connections, one can formally write 
\begin{equation}\label{E:sept1710_B}
I_0 := \frac{\partial \,\mathrm{vec}\,  \left(C_0\right)  }{\partial \, \mathrm{vec}\,  \left(C_0 \right)  } =   \left [ \begin{array}{cc} \widehat{I_K} & 0 \\ 0 & \widehat{I_{L^t}} \end{array} \right ]. 
\end{equation}

Replacing $C$ with $C_0$ in the definition of $F$ and repeating the whole process of finding the Jacobian of $F$, now with respect to  $C_0$, and using Eq.~(\ref{E:sept1710_B}) we obtain the new Jacobian \[
J_0 :=  \frac{\partial \,\mathrm{vec}\,  (F\left(C_0\right) ) }{\partial \, \mathrm{vec}\,  \left(C_0 \right) }
= I_0 J I_0, 
\]
where the matrix $J$ is the Jacobian matrix  given in Eq.~(\ref{1april10_jacobian}). 
If the connectivity matrices have zero entries, then $I_0$ will have zero entries in the diagonal. Therefore, some eigenvalues of $J_0$ will be zero. But, this does not affect the stability of slow manifold because we only need to look for the stability along the directions of intermediate complexes that occur in the reactions. 
That is, we only need to look at the principal submatrix of $J_0$ corresponding to the positive entries in the diagonal of $I_0$. Let this principal submatrix be $J_0^+$. But, since $I_0 J_0 I_0 = I_0 J I_0$, we see that $J_0^+$ is also a principal submatrix of $J$. And $J_0^+$ is  independent of   zero entries in the connectivity matrices. Since Lemma~\ref{15mar10_stableMatrix} implies that, when all the entries in connectivity matrices are positive, $J$ has eigenvalues with only negative real parts, we get that  $J_0^+$ will have eigenvalues with only negative real parts.  We conclude that the results hold even if some entries in the connectivity matrices are zero.









\end{document}